\documentclass[aps,prb,twocolumn,groupedaddress,superscriptaddress]{revtex4-1}

\usepackage{graphicx}
\usepackage{latexsym}
\usepackage{amssymb}
\usepackage{amsmath}
\usepackage{amsthm}
\usepackage{amsfonts}
\usepackage{upgreek}
\usepackage{bm,dsfont}
\usepackage{multirow}
\usepackage{color}
\usepackage{footnote}
\usepackage{amsmath}
\usepackage{bbm}
\usepackage{bm}
\usepackage{natbib}
\setcitestyle{square}
\usepackage{siunitx}
\usepackage{float}
\usepackage{bbold}
\usepackage[normalem]{ulem}

\usepackage{tikz-cd} 
\usetikzlibrary{decorations.pathmorphing}

\usepackage{hyperref}
\hypersetup{
colorlinks = true,
linkcolor = [rgb]{0.70,0.13,0.13},
citecolor = [rgb]{0.13,0.55,0.13},
urlcolor = [rgb]{0.25, 0.41, 0.88}}
\def\ra{\rangle}
\def\la{\langle}

\def\mc{\mathcal}

\def\bea{\begin{eqnarray}}
\def\eea{\end{eqnarray}}
\def\be{\begin{equation}}
\def\ee{\end{equation}}
\def\beal{\begin{aligned}}
\def\eeal{\end{aligned}}
\def\bew{\begin{widetext}}
\def\eew{\end{widetext}}

\def\bea{\begin{array}}
\def\eea{\end{array}}

\def\ra{\rangle}

\def\b{\boldsymbol}

\def\mcP{\mathcal{P}}

\DeclareRobustCommand{\rchi}{{\mathpalette\irchi\relax}}
\newcommand{\irchi}[2]{\raisebox{\depth}{$#1\chi$}} 

\newcommand{\rmi}{\mathrm{i}}
\newcommand{\rme}{\mathrm{e}}
\newcommand{\bs}{\boldsymbol}
\newcommand{\da}{\dagger}

\newcommand{\bra}[1]{\langle #1|}
\newcommand{\ket}[1]{|#1 \rangle}
\newcommand{\braket}[2]{\langle #1|#2 \rangle}
\newcommand{\inner}[2]{\langle #1|#2\rangle}
\newcommand{\ii}{\mathrm{i}}
\newcommand{\id}{\mathds{1}}

\newcommand{\scP}{\mathcal{P}}
\newcommand{\scQ}{\mathcal{Q}}
\newcommand{\scH}{\mathcal{H}}

\newcommand{\scM}{\mathcal{M}}

\newcommand{\vac}{\mathsf{vac}}
\newcommand{\Tr}{\operatorname{Tr}}

\renewcommand{\Re}{\operatorname{Re}}
\renewcommand{\Im}{\operatorname{Im}}
\newcommand{\vect}[1]{{\bm{#1}}}

\newcommand{\dia}[3]{\raisebox{#3pt}{\includegraphics[height=#2pt]{dia_#1}}}
\newcommand{\eq}[1]{\begin{equation}#1\end{equation}}
\newcommand{\eqs}[1]{\begin{equation}\begin{split}#1\end{split}\end{equation}}
\newcommand{\eqnref}[1]{Eq.\,\eqref{#1}}
\newcommand{\figref}[1]{Fig.\,\ref{#1}}

\newcommand{\secref}[1]{Sec.\,\ref{#1}}\newcommand{\appref}[1]{Appendix\,\ref{#1}}

\newtheorem{theorem}{Theorem}[section]
\newtheorem{lemma}{Lemma}[section]
\newtheorem{corollary}{Corollary}[section]
\newtheorem{definition}{Definition}[section]

\begin{document}


\title{Quantum Magnetism in Wannier-Obstructed Mott Insulators}

\author{Xiao-Yang Huang}
\affiliation{Department of Physics, University of California, San Diego, CA 92093, USA}
\affiliation{School of Science, Xi'an Jiaotong University, Xi'an, Shaanxi 710049, China}

\author{Taige Wang}
\affiliation{Department of Physics, University of California, San Diego, CA 92093, USA}

\author{Shang Liu}
\affiliation{Department of Physics, Harvard University, Cambridge, MA 02138, USA}

\author{Hong-Ye Hu}
\affiliation{Department of Physics, University of California, San Diego, CA 92093, USA}

\author{Yi-Zhuang You}
\affiliation{Department of Physics, University of California, San Diego, CA 92093, USA}

\begin{abstract}

We develop a strong coupling approach towards quantum magnetism in Mott insulators for Wannier obstructed bands. Despite the lack of Wannier orbitals, electrons can still singly occupy a set of exponentially-localized but nonorthogonal orbitals to minimize the repulsive interaction energy. We develop a systematic method to establish an effective spin model from the electron Hamiltonian using a diagrammatic approach. The nonorthogonality of the Mott basis gives rise to multiple new channels of spin-exchange (or permutation) interactions beyond Hartree-Fock and superexchange terms. We apply this approach to a Kagome lattice model of interacting electrons in Wannier obstructed bands (including both Chern bands and fragile topological bands). Due to the orbital nonorthogonality, as parameterized by the nearest neighbor orbital overlap $g$, this model exhibits stable ferromagnetism up to a finite bandwidth $W\sim U g$, where $U$ is the interaction strength. This provides an explanation for the experimentally observed robust ferromagnetism in Wannier obstructed bands. The effective spin model constructed through our approach also opens up the possibility for frustrated quantum magnetism around the ferromagnet-antiferromagnet crossover in Wannier obstructed bands.

\end{abstract}

\maketitle

\section{Introduction}
Mott insulators are correlated insulators where electrons singly occupy localized orbitals to avoid the repulsive interaction. In many cases, Mott insulators further develop antiferromagnetic ordering below the charge gap due to the superexchange interaction among low-energy spin degrees of freedoms. However, in recent twisted bilayer graphene (tBLG) experiments\cite{Cao2018c,Cao2018d}, the observation of ferromagnetic hysteresis \cite{ISI:000483195200043} suggests that the three-quarter-filling Mott insulating state exhibits ferromagnetism. In another experiment on twisted double bilayer graphene (tDBLG) \cite{liu2019spinpolarized,shen2019observation,cao2019electric}, an increasing gap under in-plane magnetic field also suggests a ferromagnetic phase. Such ferromagnetic Mott state is understood in the flat band limit\cite{Mielke_1992,Mielke&Tasaki,Tasaki1996,ISI:000073659300001}, where the spin exchange term in Coulomb interaction reduces energy of the spin-polarized state. For similar reasons, the ferromagnetic phase in Moir\'e superlattice systems has a large overlap with the flat band ferromagnetism.\cite{alavirad2019ferromagnetism,2019arXiv190108110B,senthil2019narrow} Faithful treatments of correlated Moir\'e superlattice systems have been proposed in various ways, including projecting Coulomb interaction onto the symmetry-broken (or obstruction-free) Wannier basis\cite{PhysRevLett.122.246402,PhysRevLett.123.096802,PhysRevLett.122.246401,PhysRevB.99.205150,PhysRevResearch.1.033126} and analyzing the interacting effects within momentum space in the weak-coupling limit.\cite{PhysRevX.8.031089,You2019Superconductivity,repellin2019ferromagnetism,PhysRevB.99.075127,Lee2019,xie2018nature,liu2019nematic,bultinck2019ground} While most analyses consistently point to a robust ferromagnetism near the flat-band limit, it remains challenging to analyze the instability of such ferromagnetic state in competition with the antiferromagnetic superexchange away from the flat-band limit. 

One major obstacle to model the competition between exchange and superexchange effects systematically in tBLG has to do with the Wannier obstruction,\cite{PhysRevLett.121.126402,PhysRevB.99.125122} i.e.~if Wannier orbitals are constructed with only the relevant bands, they cannot respect all the symmetries.\cite{PhysRevX.8.031088,PhysRevB.98.045103} To construct the Wannier orbitals for tBLG with all symmetries taken into account, the minimal model needs to contain at least ten bands as pointed out by H. C. Po \textit{et al}.\cite{PhysRevX.8.031089,PhysRevB.99.195455,PhysRevB.98.085435} If we only focus on a few bands within the energy scale of interaction, the Wannier obstruction would prevent us from constructing Wannier orbitals. In lack of Wannier orbitals, it becomes unclear how the electrons should localize to form Mott insulators, which further obscure the derivation of low-energy effective spin (and/or valley) models. The issue of Wannier obstruction has long been identified and studied in quantum Hall systems and other Chern insulators. In the phases with nonzero Chern number, exponentially-localized Wannier orbitals cannot be constructed regardless any symmetry considerations.\cite{PhysRevLett.98.046402} Now with more fragile topological systems being discovered,\cite{FragileExpt,Bernevig,2018arXiv181111786L} the Wannier obstruction becomes a more general issue in the study of strongly correlated electronic systems.

We begin our general discussion by considering an extended Hubbard model on an arbitrary lattice,
\eqs{\label{eq:H}H&=H_{t}+H_{U},\\
H_{t}&=\sum_{ij}t_{ij}c_i^\dagger c_j,\\
H_{U}&= \frac12 \sum_{ij}U_{ij}:n_i n_j:.}
where $c_i=(c_{i,\uparrow},c_{i,\downarrow})^\intercal$ is the electron operator containing spin degrees of freedom and $n_i=c_i^\dagger c_i$ is the total electron number on site $i$. The model may be generalized to include orbital or valley\cite{PhysRevB.99.195455} degrees of freedom, but in this work we will only focus on spins for illustration purpose. We assume certain degree of locality in $t_{ij}$ and $U_{ij}$. Suppose that $H_{t}$ produces several bands, described by the dispersion relations $\epsilon_n(\vect{k})$, 
\eq{\label{eq:Ht}
H_{t}=\sum_{n\vect{k}}c_{n\vect{k}}^\dagger \epsilon_n(\vect{k})c_{n\vect{k}}.}
In many cases, we are only interested in a subset of these bands near the Fermi energy with total bandwidth $W$, and well separated from other bands by energy gap $\Delta$. Although the full band structure has a tight-binding model description, a subset of these bands may not. H. C. Po \textit{et al.}\cite{PhysRevLett.121.126402,PhysRevB.99.125122} discussed several scenarios for a subset of bands, which can be trivial, obstructed trivial, fragile topological, or stably topological. 
As long as these bands are not trivial, they admit the Wannier obstruction, which is an obstruction towards constructing symmetric, exponentially-localized and orthogonal Wannier orbitals by linearly combining Bloch states within the subset of bands.
The well known ones are with Chern number where no additional symmetry is required.\cite{PhysRevLett.49.405,PhysRevLett.98.046402,PhysRevLett.107.126803} We also have fragile topology where they are obstructed but the obstruction can be removed by adding trivial bands (essentially site localized bands) below the Fermi energy which can then be mixed with the bands of interest to obtain Wannier orbitals.\cite{FragileExpt,Bernevig,2018arXiv181111786L} Finally, for the obstructed trivial bands, the obstruction is only due to the fact that the Wannier center does not reside on a lattice site, which can be resolved by adding empty sites.
The absence of such Wannier orbitals prevents us from writing down an effective tight-binding model targeting only those subset of bands of interest. This is precisely the case for tBLG, where the relevant conduction and valance bands are fragile topological and hence Wannier obstructed by the $C_2\mc T$ symmetry.\cite{PhysRevX.8.031089,PhysRevB.99.195455,PhysRevB.98.085435} 

Weak-coupling approaches have been developed to treat the interaction perturbatively (as $U\ll W$),\cite{PhysRevX.8.031089,You2019Superconductivity,Lee2019} such that one only need to work with a momentum space description of the effective band structure near the Fermi surface, hence circumventing the Wannier obstruction. However, it remains challenging to understand the strong-coupling physics in Wannier obstructed bands, when the energy scales are arranged in the following hierarchy
\eq{W\ll U\ll \Delta.}
The band gap $\Delta$, as the leading energy scale, protects a set of Wannier obstructed low-energy bands from mixing with high-energy bands away from the Fermi surface. To further respect the interaction energy $U$, electrons should repel each other into real-space localized orbitals by combining states in the low-energy bands. In the standard notion of Mott insulators, electrons are localized in Wannier orbitals with charge fluctuation gapped and spin fluctuation remained active at low energy. Now in the absence of Wannier orbitals, does the many-body Hamiltonian in \eqnref{eq:H} still admits Mott-like ground states? If so, can we write down the trial wave function to describe the low-lying states? Can we derive an effective model to describe the spin dynamics at low energy?

Motivated by these questions, we take a closer look at the requirements of Wannier orbitals, namely symmetry, locality and orthogonality. If we sacrifice one or more of them, the Wannier obstruction can be lifted and it would be possible to construct orbitals to host electrons in a Mott-like state. If we sacrifice the symmetry requirement, we will have to fine tune hopping parameters to fit the band structure. If we sacrifice the locality requirement, we will end up with a non-local hopping model which is hard to deal with. So we decided to explore the possibility of sacrificing orthogonality and working with a set of nonorthogonal Wannier basis. This approach is in analogous to the Maki-Zotos wavefunction \cite{PhysRevB.28.4349} and the von Neumann lattice formulation \cite{PhysRevB.42.10610,PhysRevB.51.5048,PhysRevB.65.075311} in quantum Hall systems with nonzero but negligible orbital overlap.

We present the general theory based on nonorthogonal Wannier basis with finite orbital overlap in \secref{sec:theory}. The nonorthogonality of the orbitals leads to new spin exchange channels and chiral spin exchange channels. The competition among these new channels could lead to a rich magnetic phase diagram. We discuss leading order contributions to the effective spin Hamiltonian in \secref{sec:spin} and discuss several new channels emerging from the theory of nonorthogonal basis. We also propose an energetic objective function in \secref{construct} to construct these nonorthogonal orbitals numerically to facilitate the study of any concrete model. Within this framework we study a toy model proposed in Ref.~\onlinecite{PhysRevB.99.125122} in \secref{model} to demonstrate our framework in Wannier-obstructed bands. In this model, we show that the flat-band ferromagnetism remains stable up to finite band width and a variety of magnetic phases appear around the the ferromagnet-antiferromagnet crossover.

\section{Theory of nonorthogonal Basis} \label{sec:theory}

In this section, we discuss how to project the Hamiltonian in \eqnref{eq:H} to a nonorthogonal spin basis of low-energy states and how to treat perturbative corrections. Let us assume a set of localized and normalized but nonorthogonal orbitals $\phi_{I}(i)$ in the real space labeled by the orbital index $I$, which jointly labels the unit cell and the orbital within the unit cell. We will leave the energetic criterion to optimize these orbitals and their completeness as a set of basis for later discussions in \secref{construct}. For now, we assume that electrons will self-organize under repulsive interaction to develop such nonorthogonal localized orbitals. The nonorthogonality implies a non-trivial metric $g_{IJ}$ among these orbitals,
\begin{equation} \label{eq:g}
    g_{IJ} \equiv \sum_i\phi_{I}^*(i)\phi_J(i) \neq \delta_{IJ}.
\end{equation}
Nevertheless, we assume the normalization condition $g_{II}=1$ for all $I$. Given these orbitals, we can define a set of fermion operators $a_{I\sigma}^\dagger$ that create electrons residing on these orbitals,
\eq{a^\da_{I\sigma}=\sum_{i}\phi_{I}(i)c^\da_{i\sigma},\quad(\sigma=\uparrow,\downarrow)}
such that they satisfy the following anticommutation relations $\{a_{I\sigma},a_{J\tau}\} = \{a_{I\sigma}^\dagger,a_{J\tau}^\dagger\} = 0$ and $\{a_{I\sigma},a_{J\tau}^\dagger\} = g_{IJ} \delta_{\sigma\tau}$. These operators $a_{I\sigma}^\dagger$ allow us to construct a set of many-body trial states from the vacuum state $\ket{\vac}$,
\eq{\label{eq:ketsigma}\ket{\Psi_\vect{\sigma}}=\prod_{I}a_{I\sigma_I}^\dagger\ket{\vac}.}
In these many-body states, every orbital is singly occupied and the spin configuration is labelled by $\vect{\sigma}=\{\sigma_I\}$. The nonorthogonality of single particle orbitals also implies the nonorthogonality of these many-body states $\braket{\Psi_\vect{\sigma}}{\Psi_\vect{\tau}} \neq \delta_{\vect{\sigma}\vect{\tau}}$.
In the Mott limit $U \gg W$, the trial states $\ket{\Psi_\vect{\sigma}}$ span the low energy manifold $\scH$. All the charge degrees of freedom are frozen, and the spin degrees of freedom are still allowed to fluctuate, resembling the Mott states. We can then derive the effective theory for these spin degrees of freedom and investigate the resulting phases.

\subsection{Exchange Interactions and Beyond}

To formulate an effective spin model, we first introduce the spin Hilbert space $\tilde{\scH}$ spanned by a fictitious set of \emph{orthogonal} Ising basis $\ket{\vect{\sigma}}$, which allows us to define the spin operator $\vect{S}_I=(S_I^x,S_I^y,S_I^z)$ in the conventional way $\bra{\vect{\sigma}}S_I^a\ket{\vect{\tau}}=\frac{1}{2}\sigma_{\tau_I\sigma_I}^a\prod_{J\neq I}\delta_{\tau_J\sigma_J}$, where $\sigma^a_{\tau\sigma}$ denotes the Pauli matrix element. We would like to comment that the Ising states $\ket{\vect{\sigma}}$ do not directly correspond to physical electronic states, but merely play a bookkeeping role to provide a convenient basis for the purpose of representing spin operators. The relation between different Hilbert spaces is illustrated in \figref{fig:spaces}.

\begin{figure}[htbp]
\begin{center}
\includegraphics[width=0.9\columnwidth]{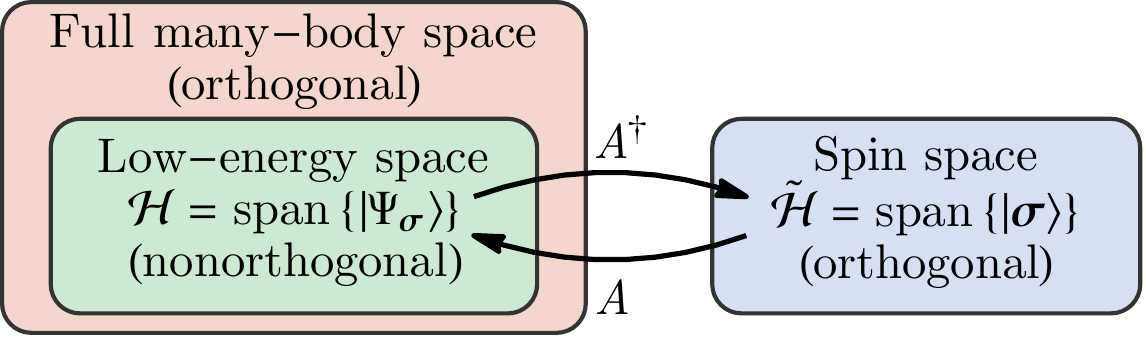}
\caption{The lattice model in \eqnref{eq:H} is defined in the full many-body Hilbert space spanned by the orthogonal Fock states of electrons, which includes a low-energy subspace $\scH$ spanned by the nonorthogonal trial states $\ket{\Psi_\vect{\sigma}}$ in \eqnref{eq:ketsigma}. A spin Hilbert space $\tilde{\scH}$ spanned by the orthogonal Ising basis $\ket{\vect{\sigma}}$ is introduced to represent the effective spin model. The linear map $A$ (and $A^\dagger$) connects $\scH$ and $\tilde{\scH}$.}
\label{fig:spaces}
\end{center}
\end{figure}

We want to project the many-body Hamiltonian \eqnref{eq:H} to the low energy subspace $\scH$ of electrons and then translate it to the spin space $\tilde{\mathcal{H}}$. The complication arises from the nonorthogonality of the many-body trial states $\ket{\Psi_\vect{\sigma}}$ that span the low energy subspace $\scH$. Here we present a systematic approach to deal with the nonorthogonality. First we introduce a non-unitary linear transformation $A : \tilde{\mathcal{H}} \to \mathcal{H}$ to map the Ising basis $\ket{\vect{\sigma}}$ to the many-body trial states $\ket{\Psi_\vect{\sigma}}$ with corresponding spin configuration,
\begin{equation}
    A=\sum_{\vect{\sigma}}\ket{\Psi_\vect{\sigma}}\bra{\vect{\sigma}}.
\end{equation}
Next we define the adjoint operator $A^{\dagger} = \sum_{\vect{\sigma}}\ket{\vect{\sigma}}\bra{\Psi_\vect{\sigma}}$, such that $\braket{\Psi_\vect{\sigma}}{A \vect{\tau}} = \braket{A^{\dagger} \Psi_\vect{\sigma}}{\vect{\tau}}$. Note that $A$ is not a unitary transformation, so $A^\dagger\neq A^{-1}$. In fact, $A^{-1} = \sum_{\vect{\sigma}}\ket{\vect{\sigma}}\bra{\bar{\Psi}_\vect{\sigma}}$, where $\{\ket{\bar{\Psi}_\vect{\sigma}}\}$ is the dual basis to $\{\ket{\Psi_\vect{\sigma}}\}$ such that $\la\bar{\Psi}_\vect{\sigma}|\Psi_\vect{\tau}\ra=\delta_{\vect{\sigma\tau}}$.  Now we can use the operators $A$ and $A^{\dagger}$ to project the identity operator $\id$ and Hamiltonian $H$ to the Ising basis,
\begin{equation} \label{eq:transformation}
    \begin{aligned}
        \id &\to A^{\dagger}A\equiv G,\\
        H &\to A^{\dagger} H A \equiv \tilde{H}.
    \end{aligned}
\end{equation}
Under the projection, the original eigen problem $H\ket{\Psi}=E\ket{\Psi}$ becomes a generalized eigen problem by inserting the identity operator $A A^{-1}$ and multipling $A^{\dagger}$ from the left,
\begin{equation} \label{eq:general}
    \tilde{H}\ket{\Phi}=EG\ket{\Phi},
\end{equation}
where $\ket{\Phi} \equiv A^{-1}\ket{\Psi}$ is the representation of the many-body eigenstate $\ket{\Psi}$ in the Ising basis. Now the generalized eigen problem in \eqnref{eq:general} is formulated in the spin Hilbert space $\tilde{\scH}$ with a nice orthogonal basis.

After the projection, we can expand the many-body Gram matrix $G$ and the projected Hamiltonian $\tilde{H}$ as linear combinations of permutation operators $\rchi_\scP$ in the spin space,
\eqs{\label{eq:GH formula}
G&=\sum_{\scP\in S_N}(-)^\scP G_\scP\rchi_\scP,\\
\tilde{H}&=\sum_{\scP\in S_N}(-)^\scP H_\scP\rchi_\scP,}
where $S_N$ denotes the permutation group over all orbitals $I$, $(-)^\scP$ is the sign of the permutation $\scP$, and $\rchi_\scP$ is the permutation operator that permutes the spins among different orbitals $\rchi_\scP \ket{\{\sigma_I\}} = \ket{\{\sigma_{\scP^{-1}(I)}\}}$. For two-spin and three-spin permutations, $\rchi_\scP$ can be expressed with the familiar Heisenberg term and chiral spin term,
\begin{equation} \label{eq:permutation}
    \begin{aligned}
        \rchi_{(IJ)} &= \tfrac{1}{2}+2\vect{S}_I\cdot\vect{S}_J,\\
        \rchi_{(IJK)} &= \tfrac{1}{4}+\vect{S}_I\cdot\vect{S}_J+\vect{S}_J\cdot\vect{S}_K+\vect{S}_K\cdot\vect{S}_I\\
        &\phantom{=} - 2\ii\vect{S}_I\cdot(\vect{S}_J\times\vect{S}_K). 
    \end{aligned}
\end{equation}
To specify the coefficients $G_\scP$ and $H_\scP$, we introduce the hopping tensor $t_{IJ}$ and the interaction tensor $U_{IJKL}$,
\eqs{\label{eq:matrix}
t_{IJ}&=\sum_{ij}\phi_I^*(i)t_{ij}\phi_J(j),\\
U_{IJKL}&=\sum_{ij}\phi_I^*(i)\phi_J(i)U_{ij}\phi_K^*(j)\phi_L(j),}
where $t_{ij}$ and $U_{ij}$ are the bare hopping and interaction coefficients in $H$, as introduced in \eqnref{eq:H}. Given the tensors $g_{IJ}$ in \eqnref{eq:g} and $t_{IJ},U_{IJKL}$ in \eqnref{eq:matrix}, the coefficients $G_\scP$, $H_\scP\equiv T_\scP +U_\scP$ in \eqnref{eq:GH formula} are given by
\eqs{\label{eq:co}G_\scP&=\prod_{I}g_{\scP(I)I},\\
T_\scP&= \sum_{I}t_{\scP(I)I}\prod_{J\neq I}g_{\scP(J)J},\\
U_\scP&= \frac{1}{2}\sum_{I\neq J}U_{\scP(I)I\scP(J)J}\prod_{K\neq I,J}g_{\scP(K)K},}
where $T_\scP$ and $U_\scP$ denotes the contribution from the hopping and the interaction terms respectively. 

To simplify the notation, we introduce the diagrammatic representation: a closed loop of arrows represents a permutation cycle among the orbitals. If an arrow from $I$ to $J$ is labeled by $t$ or $g$, it contributes a factor of $t_{IJ}$ or $g_{IJ}$; if two arrows, e.g. one from $I$ to $J$ and the other from $K$ to $L$, are connected and labeled by $U$, it contributes a factor of $U_{IJKL}$. For example,
\begin{equation}
    \dia{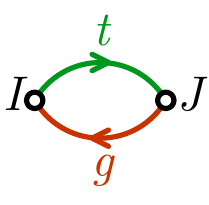}{34}{-17}\equiv t_{IJ}g_{JI}, \quad
    \dia{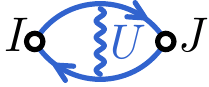}{14}{-7} \equiv U_{IJJI}.
\end{equation}
Using the diagrammatic representations, one can expand $G$ and $\tilde{H}$ in terms of spin permutations $\rchi_\scP$ order by order. To the order of two-spin exchange, we have 
\begin{equation}\label{eq:GH}
\begin{split}
G&=\id-\sum_{(IJ)}\Big(\dia{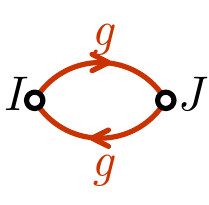}{34}{-15}\Big)\rchi_{(IJ)}+\cdots,\\
\tilde{H}&=\sum_{I}\Big(\id-\sum_{(KL)}\Big(\dia{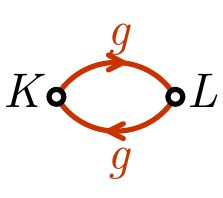}{34}{-15}\Big)\rchi_{(KL)}\Big)\\
&\hspace{36pt}\times\Big(\dia{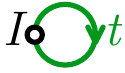}{13}{-5}+\frac{1}{2}\sum_{J}\dia{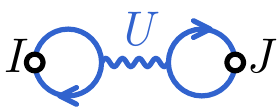}{18}{-6}\Big)\\
&-\sum_{(IJ)}\Big(\dia{tg}{34}{-15}+\frac{1}{2}\dia{UF}{14}{-5}+\frac{1}{2}\sum_{K}\dia{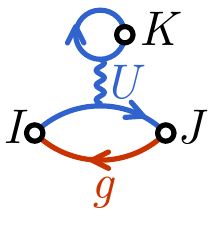}{37}{-15}\Big)\rchi_{(IJ)}\\
&+\cdots,
\end{split}
\end{equation}
where the permutation operator $\rchi_{(IJ)}$ can eventually be written in terms of spin operators as in \eqnref{eq:permutation}. Higher order permutations can be included systematically based on \eqnref{eq:GH formula}. For all summations appeared in \eqnref{eq:GH}, it is assumed that the orbitals labeled by different indices do not coincide. In the orthogonal limit $g_{IJ} = \delta_{IJ}$, all diagrams in \eqnref{eq:GH} that contains $g$-labeled (red) arrows will vanish, such that $G$ reduces back to the identity operator $\id$ and $\tilde{H}$ reduces to the following three terms
\begin{equation}\label{eq:HartreeFock}
\tilde{H}=\sum_{I}\dia{t}{13}{-5}+\frac{1}{2}\sum_{I\neq J}\Big(\dia{UH}{18}{-6}-\dia{UF}{14}{-5}\;\rchi_{(IJ)}\Big),
\end{equation}
The first diagram is the band energy. The second and third diagrams are respectively the Hartree and the Fock energies between electrons from orbitals $I$ and $J$ (where the Fock interaction is accompanied with the spin exchange $\rchi_{(IJ)}$). For nonorthogonal orbitals, the metric $g_{IJ}$ becomes non-trivial, then non-vanishing terms (containing $g$-labeled arrows) in \eqnref{eq:GH} suggest contributions beyond Hartree-Fock approximation, which lead to new channels of spin exchange interactions. Furthermore, there exist three- and even more spin interactions that are conventionally only present through the interaction-suppressed super-exchange effects. The various spin permutation interactions competing with each other could result in a frustrated quantum magnet with a rather rich phase diagram. To determine the ground state of the low-energy spin degrees of freedom in Wannier-obstructed Mott insulators, one will need to solve the generalized eigen problem in \eqnref{eq:general}.

\subsection{Superexchange Interactions from Perturbation} \label{sec:perturbation}

In the above discussion, we project the many-body Hamiltonian $H$ to the low-energy subspace $\scH$ spanned by the Mott states $\ket{\Psi_\vect{\sigma}}$ to derive the effective spin model in the flat-band limit  $W\ll U$. Away from the flat-band limit, the electrons can virtually hop to the neighboring orbitals and back, which gives rise to the superexchange interactions among the spins. To capture such effect, we should go beyond the low-energy subspace $\scH$ and consider the perturbation effects in orders of $(t/U)$.

In the following, we analyze the perturbative corrections to the effective spin model within our framework. In \appref{app:perturbation}, we review the generalized perturbation theory for nonorthogonal basis. We apply the perturbation theory to the spin dynamics in the low energy manifold $\mathcal{H}$ by treating the hopping Hamiltonian $H_t$ in \eqnref{eq:Ht} as a small perturbation.\footnote{Strictly speaking, we should treat the bandwidth $W$ as perturbation, such that $H_t$ corresponds to the difference $H_W - H_{W=0}$, where $H_{W=0}$ corresponds to the band-flattened version of the band structure, but this will not affect our formulation of the general approach.} Consider the high energy subspace spanned by states $\ket{n,\alpha}$ with double occupancy but still the same number of electrons, where $n$ labels the number of doubly occupied orbitals and $\alpha$ labels the configuration. Assuming that the zeroth order energy is completely determined by $n$ and states with different $n$'s are orthogonal, i.e. $\braket{n,\alpha}{m,\beta} = \delta_{mn} G_{n\alpha\beta}$, we get the following energy correction from the second-order perturbation theory
\begin{equation} \label{eq:pert2}
    \bra{\Psi_{\vect{\sigma}}}H^{(2)}\ket{\Psi_{\vect{\tau}}} = - \sum_{n > 0} \frac{\bra{\Psi_{\vect{\sigma}}}H_t\ket{n,\alpha}G_{n}^{\alpha \beta}\bra{n,\beta}H_t\ket{\Psi_{\vect{\tau}}}}{E_n - E_0}
\end{equation}
where $G_{n}^{\alpha \beta}$ is the inverse of $G_{n\alpha \beta}$. Since the dominant contribution of $\bra{\Psi_{\vect{\sigma}}}H_t\ket{n,\alpha}$ comes from states with only one doubly occupied site, we can approximate \eqnref{eq:pert2} by
\begin{equation}
    \bra{\Psi_{\vect{\sigma}}}H^{(2)}\ket{\Psi_{\vect{\tau}}} \approx - \frac1U \left ( \bra{\Psi_{\vect{\sigma}}}H_t^2 - H_t \id_{\Psi} H_t \ket{\Psi_{\vect{\tau}}} \right )
\end{equation}
where $U$ is the energy required to create one double occupancy, and $\id_{\Psi}=\sum_{\vect{\sigma}}\ket{\Psi_\vect{\sigma}}\bra{\bar{\Psi}_\vect{\sigma}}$ is the projection operator to the low energy subspace. Now we can use the projection introduced previously to write down the Hamiltonian $\tilde{H}^{(2)} \equiv A^{\dagger} H^{(2)} A$ in the Ising basis,
\begin{equation} \label{eq:pertco}
    \tilde{H}^{(2)} = - \frac{1}{U} \left ( \widetilde{H_t^2} - \tilde{H}_t G^{-1} \tilde{H}_t \right ),
\end{equation}
where $\widetilde{H_t^2} \equiv A^{\dagger} H_t^2 A$ follows our convention, and in the second term we use $A A^{-1} = \left ( A^{\dagger} \right )^{-1} A^{\dagger} = \id_{\Psi}$. We can further write $\tilde{H}^{(2)}$ in terms of hopping tensors,
\begin{equation}
    \begin{split}
    \tilde{H}^{(2)} & = - \frac{1}{U} \Big(\sum_{\scP\in S_N}(-)^\scP \left ( T^2 \right )_\scP \rchi_\scP \\ 
    & \hspace{36pt}- \sum_{\scP\circ\scQ \in S_N}(-)^{\scP\circ\scQ} T_\scP G^{-1} T_{\scQ} \rchi_{\scP \circ \scQ} \Big) ,
\end{split} 
\end{equation}
where $\circ$ denotes the composition of permutations (such that $\rchi_{\scP\circ\scQ}=\rchi_{\scP}\rchi_{\scQ}$), and the coefficients $T_\scP$ and $\left ( T^2 \right )_\scP$ are given by 
\begin{equation} \label{eq:H2formula}
    \begin{aligned}
        T_\scP &= \sum_{I}t_{\scP(I)I}\prod_{J\neq I}g_{\scP(J)J},\\
        \left ( T^2 \right )_\scP &= \sum_{I}\left (t^2 \right)_{\scP(I)I}\prod_{J\neq I}g_{\scP(J)J}\\
        &\phantom{=}+\sum_{I\neq J}t_{\scP(I)I}t_{\scP(J)J}\prod_{K\neq I,J}g_{\scP(K)K},
    \end{aligned}
\end{equation}
and $\left( t^2 \right)_{IJ}=\sum_{ijk}\phi_I^*(i)t_{ik}t_{kj}\phi_J(j)$. To the order of two-spin exchange, we have
\eqs{\label{eq:H2}
 \tilde{H}^{(2)} &= \frac{1}{U} \sum_{I}\Big(\id-\sum_{(JK)}\Big(\dia{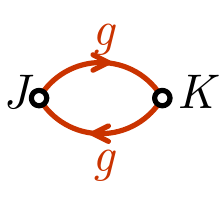}{34}{-15}\Big)\rchi_{(JK)}\Big)\\
 &
\hspace{46pt}\times\Big(\dia{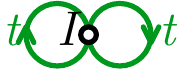}{13}{-4}-\dia{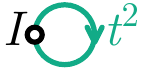}{13}{-5}\Big)\\
 &\hspace{12pt}-\frac{1}{U}\sum_{(IJ)} \Big(2\Big(\dia{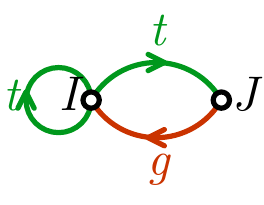}{34}{-15}+\dia{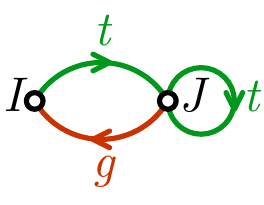}{34}{-15}\Big)\\
&\hspace{46pt}-\dia{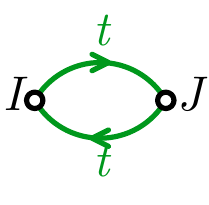}{34}{-15}-\dia{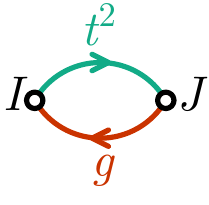}{34}{-15}\Big) \rchi_{(IJ)} 
+\cdots}
For all summations appeared in \eqnref{eq:H2}, it is assumed that orbitals labeled by different indices do not coincide. In the orthogonal limit $g_{IJ} = \delta_{IJ}$, all diagrams containing $g$-labeled arrows will vanish, such that the second-order perturbation
\eq{
 \tilde{H}^{(2)} =\frac{1}{U}\sum_{(IJ)} \Big(\dia{tt}{34}{-15}\Big) \rchi_{(IJ)} +\text{const.}}
contains only the usual $t^2/U$ antiferromagnetic superexchange interaction. When we allow non-trivial $g_{IJ}$, both $T_\scP$ and $T_\scP^2$ can give rise to new channels in spin interactions. Mediated by $g_{IJ}$, three- or higher order spin interactions can also arise even at the level of second-order perturbation in $(t/U)$. Similar treatment can be generalized to higher order perturbation theory.

\section{Effective Spin Hamiltonian} \label{sec:spin}

In this section, we discuss how to solve the effective spin model constructed in \secref{sec:theory}. There are two major challenges. First, the model is presented as a generalized eigenvalue problem in \eqnref{eq:general}, which requires to diagonalize a complicated operator $G^{-1} \tilde{H}$ that is not guaranteed to be short-ranged on the lattice. Second, the summation over all permutation in \eqnref{eq:GH formula} is hard to track even numerically. Both challenges can be resolved by separating connected permutations and disconnected permutations. A connected permutation is formally defined as a cyclic permutation (or cycle) in group theory, while those that are not cycles are called disconnected in the following text. As demonstrated in \eqnref{eq:GH}, diagrams in $\tilde{H}$ can be organized by the connected diagram (containing $t$ or $U$), each followed by a series of disconnected diagrams (containing $g$ only). The series of disconnected diagrams is similar to $G$, which motivates us to factor $G$ out of $\tilde{H}$. However, residue terms are generated due to over counting diagrams with colliding indices, illustrated as follows
\begin{equation}\label{eq:H=Gexpand}
\tilde{H}=G\sum\Big(\dia{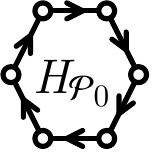}{26}{-12}-\dia{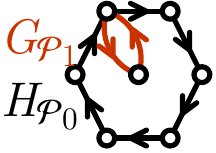}{26}{-12}+\dia{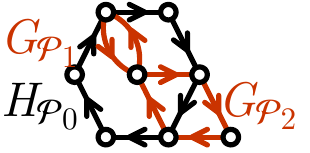}{26}{-12}+\cdots\Big),
\end{equation}
where $\scP_0$ sums over single-cycle (i.e.~connected) permutations and $\scP_{1,2,\cdots}$ sum over those permutations that have non-vanishing index overlap with every other permutation (including $\scP_0$) in the diagram. The explicit expression and a detailed convergence analysis of the entire series can be found in \appref{sec:convergence}. The rough idea is that the expansion is controlled by the small parameter $g \equiv |g_{\langle IJ \rangle}| \ll 1$ for well-localized orbitals $\phi_I$. Since $n$-spin interaction can only be generated by $\rchi_{\mcP}$ with $\text{length}(\mcP)\geq n$, they are suppressed by at least $g^{n-2}$. Thus the spin dynamics is still dominated by few-spin interactions as expected. Furthermore, for these few-spin interactions of small $n$, the contribution from sub-leading terms in the series \eqnref{eq:H=Gexpand} is further suppressed by $ng^2$ compared to the leading $g^{n-2}$ term. Thus, for those dominating few-spin interactions, we only need to consider the leading order connected diagrams:
\be
G^{-1}\tilde{H}\simeq H_c \equiv \sum_{\mcP_0\in S^*_N}(-)^{\mcP_0} H_{\mcP_0}\rchi_{\mcP_0},
\label{hc}
\ee
where $S^*_N$ is the set of single-cycle permutations in $S_N$. With this approximation, the general eigenvalue problem in \eqnref{eq:general} reduces to the ordinary eigenvalue problem
\begin{equation} \label{eq:connect}
    H_c \ket{\Phi} = E \ket{\Phi}
\end{equation}
where $H_c$ collects the connected pieces in $\tilde{H}$. The problem reduces to solving the effective spin Hamiltonian $H_c$ on a set of orthogonal basis $\ket{\vect{\sigma}}$, for which many well-developed analytical and numerical tools in quantum magnetism can be applied. Similar treatment applies to the perturbation theory described in \secref{sec:perturbation}, where the $G^{-1}$ in \eqnref{eq:pertco} has cancelled the disconnected pieces in one of the $\tilde{H}_t$'s. Then the effective spin Hamiltonian $H_c$ in \eqnref{eq:connect} gets corrected by $H_c^{(2)}$. In conclusion, despite of the nonorthogonality of the Mott basis and the complication of the generalized eigen problem, we can still work with an effective spin Hamiltonian $H_c$ in a ordinary eigen problem to  describe the low-energy spin degrees of freedoms approximately.

The full spin-rotation symmetry dictates the spin Hamiltonian $H_c$ to take the general form of
\begin{equation}
    H_c = \sum_{\langle IJ \rangle} J_{IJ} \vect{S}_I\cdot\vect{S}_J + \sum_{\langle IJK \rangle} K_{IJK} \vect{S}_I\cdot(\vect{S}_J\times\vect{S}_K) + \cdots
\end{equation}
up to three-spin interactions. Here we keep only the near neighbor interactions and analyze the coupling strengths of the Heisenberg interaction $\vect{S}_I\cdot\vect{S}_J$ and the chiral spin interaction $\vect{S}_I\cdot(\vect{S}_J\times\vect{S}_K)$.  In \eqnref{eq:GH}, among the terms attached with $\rchi_{\langle IJ \rangle}$, to the leading order of $g$, those contribute to the Heisenberg interaction are 
\begin{equation} \label{eq:2spin}
    J_{IJ} = - 4 \Re \Big ( \dia{tg}{34}{-17}+ \frac{1}{2}\dia{UF}{14}{-7} \Big).
\end{equation}
The second term is the familiar Fock exchange term, which is always positive and thus provides inter-site Hund’s coupling.
We remark that the Fock term is non-vanishing even in the orthogonal limit due to the density-density overlap between orbitals.\cite{PhysRevB.99.205150,senthil2019narrow} The first term is a new channel arising from nonorthogonality, which can be either ferromagnetic or antiferromagnetic depending on its sign. This channel could potentially provide a stronger antiferromagnetism in the strong coupling (large $U$) limit than the usual $t^2/U$ superexchange antiferromagnetism,\cite{RevModPhys.70.1039} which may enhance the magnetic frustration in the spin model. When the competition between ferromagnetism and antiferromagnetism reaches a balance in certain parameter regime, higher-order spin interactions will start to dominate the spin model. For example, there are two terms attached with $\rchi_{\langle IJK \rangle}$ that contribute to the three-spin ring exchange interaction,
\begin{equation}\label{eq:3spin}
    K_{IJK} = 2 \Im \Big (\dia{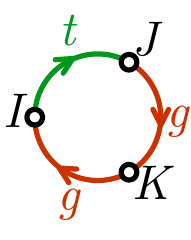}{40}{-20}+
    \dia{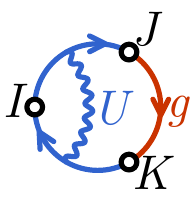}{38}{-19}+ \text{perm.} \Big ).
\end{equation}
Both terms give rise to new channels contributing to the chiral spin interaction $\vect{S}_I\cdot(\vect{S}_J\times\vect{S}_K)$. Compared to the usual $t^3/U^2$ chiral spin interaction from the 3rd order superexchange channel,\cite{PhysRevLett.116.137202,Bauer2014Chiral} these nonorthogonality enabled exchange channels could provide stronger chiral spin interaction in the strong coupling (large $U$) limit, in favor of the chiral spin liquid ground state.\cite{Wen1989Chiral}

We can repeat the analysis for the second-order perturbative correction $\tilde{H}^{(2)}$. It can also be approximated by the connected part $H^{(2)}_c$ as argued previously. If we look at the Heisenberg interaction and the chiral spin interaction,
\begin{equation}
    H^{(2)}_c = \sum_{\langle IJ \rangle} J^{(2)}_{IJ} \vect{S}_I\cdot\vect{S}_J + \sum_{\langle IJK \rangle} K^{(2)}_{IJK} \vect{S}_I\cdot(\vect{S}_J\times\vect{S}_K) + \cdots
\end{equation}
To the leading order in $g$, we get
\begin{equation}\label{eq:2nd}
\begin{split}
J^{(2)}_{IJ} &= \frac{4}{U} \Big(\dia{tt}{34}{-16}\Big),\\ 
K^{(2)}_{IJK} &= - \frac{2}{U} \Im \Big(\dia{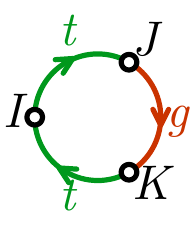}{40}{-19}+\text{perm.}\Big).
\end{split}
\end{equation}
The Heisenberg term $J_{IJ}^{(2)}$ contains contributions from the standard superexchange channel. The chiral spin term $K_{IJK}^{(2)}$ contains contributions from a new channel as nonorthogonal ring superexchange. 

Collecting all contributions from \eqnref{eq:2spin}, \eqnref{eq:3spin} and \eqnref{eq:2nd}, there are three channels that contribute to the Heisenberg interaction -- the $gt$ term from the nonorthogonal exchange, the $U$ term from the conventional exchange, and the $t^2/U$ term from the superexchange; and there are four channels that contribute to the chiral spin interaction -- the $g^2t$ and $gU$ terms from the nonorthogonal ring exchange in \eqnref{eq:3spin}, the $gt^2/U$ term from second-order perturbation theory, and the conventional $t^3/U^2$ term from third-order perturbation theory (which will appear in $H_c^{(3)}$). In the strong coupling (large $U$) limit, the novel channels originated from the orbital nonorthogonality typically dominate over the conventional superexchange and ring exchange channels. They are crucial to the analysis of the magnetism in Wannier-obstructed Mott insulators.

\section{Constructing Localized Orbitals} \label{construct}

In previous discussions, we have established the low-energy effective spin model starting from the assumption of electrons localized on a set of nonorthogonal orbitals $\phi_I(i)$. Now we come back to discuss why such arrangement is favorable and how these orbitals should be determined. In correlated materials, when the interaction energy dominates over the band width $U\gg W$, it becomes energetically favorable to recombine single-particle states in the energy band to form localized orbitals, and to arrange one electron in each localized orbital to reduce the repulsive interaction. Although the Wannier obstruction prevents us from constructing orthogonal Wannier orbitals, it does not prevent us from constructing nonorthogonal and localized orbitals, on which electrons can reside. The criterion is to minimize the total energy of the system. Thus we start from the trial many-body state $\ket{\Psi_\vect{\sigma}}$ proposed in \eqnref{eq:ketsigma}, and minimize its energy $\bra{\Psi_\vect{\sigma}}H\ket{\Psi_\vect{\sigma}}$ so as to optimize the localized orbitals $\phi_I(i)$ that were used to construct the trial state $\ket{\Psi_\vect{\sigma}}$.

However, the energy $\bra{\Psi_\vect{\sigma}}H\ket{\Psi_\vect{\sigma}}$ still depends on the spin configurations $\vect{\sigma}$, which makes the objective function undetermined. To proceed, we focus on the ``spin-independent part'' of the energy, which is naturally the constant piece in the effective spin model $H_c$. It corresponds to the Hartree energy $H_{()}$ in front of the identity operator $\rchi_{()}$, as given in \eqnref{eq:co},
\begin{equation}\label{eq:E}
\begin{split}
H_{()}&=\sum_{I}\dia{t}{13}{-6}+\frac{1}{2}\sum_{I\neq J} \dia{UH}{18}{-8}\\
&=\sum_{I}\sum_{n,ij}\phi_I^*(i)\epsilon_{n}(\ii\nabla)_{ij}\phi_I(j)\\
&\hspace{12pt}+\frac{1}{2}\sum_{I\neq J}\sum_{ij}U_{ij}|\phi_I(i)|^2|\phi_J(j)|^2.
\end{split}
\end{equation}
Here $\epsilon_{n}(\ii\nabla)$ denotes a real space representation of the band structure $\epsilon_{n}(\vect{k})$. We will assume that the $n$th band $\epsilon_{n}$ is well separated from other bands by a large band gap $\Delta$, which is much larger than the interaction strength $U$. The separating energy scales ($\Delta\gg U$) allows us to focus on the $n$th band only and find the optimal orbitals that can minimize the energy $H_{()}$. For the purpose of designing lattice models, the required separation of energy scales ($\Delta\gg W$) can be realized by applying the band flattening approach.\cite{Wang2012Fractional,Parameswaran2013Fractional,Zeng2018SUN-fractional} The energy optimization $\delta H_{()}/\delta \phi_I=0$ boils down to solving the following Gross-Pitaevskii (GP) equation,
\eq{\label{eq:GPeqn}\Big( \sum_n \epsilon_{n}(\ii\nabla)+\sum_{I\neq J} \sum_{j} U_{ij}|\phi_J(j)|^2\Big) \phi_{I}(i)=E \phi_I(i).}
The ground state can be found by imaginary time evolution of $\phi_I(i)$ under the GP Hamiltonian (i.e.~the operator on the left-hand-side of \eqnref{eq:GPeqn}). In each step of the evolution, the updated $\phi_I(i)$ orbital will be broadcasted to other unit cells by translation. During the evolution, we do not impose the orthogonality among Wannier orbitals, so the orbitals we eventually obtain are in general nonorthogonal. The interaction will automatically determine whether or not the optimal orbitals will spontaneously break the point group symmetry. 

We make a remark on the completeness of the localized orbitals. As is well known, for a single Chern band (suppose it is isolated for simplicity), it is impossible to choose a set of Bloch states $\ket{\psi_{\bs k}}$ such that it is normalized and smooth in the Brillouin zone torus.\cite{PhysRevLett.98.046402} This implies a Wannier obstruction to construct a set of exponentially localized orbitals which are orthonormal, complete and related to each other by translations. What if we relax the orthonormality constraint? Unfortunately, it is still impossible to find exponential localized orbitals which are complete and translation invariant, even if they are allowed to be nonorthogonal. If such orbitals exist, we can Fourier transform to obtain a set of unnormalized but smooth and nowhere vanishing Bloch states. Further normalizing these states leads to normalized and smooth Bloch states, causing a contradiction. It turns out that, under certain assumption which holds for the example we will be considering, it is possible to find a set of exponentially localized orbitals which are complete but nonorthogonal and also breaks the translation symmetry for one orbital. In other words, $N-1$ number of orbitals are related to each other by translations but the last orbital takes a different form, where $N$ is the total number of orbitals. We proved this claim in \appref{ChernBandWannier}. The orbital obtained from the energy optimization procedure described above, if exponentially localized, can be used to generate these $N-1$ translation related orbitals, and a different last orbital is needed to complete a basis. We expect a single orbital to have little effect on the overall physics of the system, thus we ignore this subtlety hereafter. In the appendix, we also show that the dual orbital formalism, which has been useful in the study of Hubbard model ferromagnetism in topologically trivial bands \cite{Mielke&Tasaki,Tasaki1996}, is not applicable to a Chern band with the orbitals we constructed. This is one of the reasons that we take a different approach in this work. 

\section{Application} \label{model}

\subsection{Kagome Lattice Model and Wannier Obstructions}

In this section, we apply our theory to an interacting fermion model on a Kagome lattice, whose Hamiltonian $H=H_t+H_U$ takes the form of \eqnref{eq:H}. The hopping Hamiltonian $H_t$ consists of purely imaginary nearest neighbor hopping only, with $t_{ij}=\ii/2$ for $j \to i$ along the bond direction as specified in \figref{fig:Kagome}(a). The model was introduced by Ref.~\onlinecite{PhysRevB.99.125122} to demonstrate fragile topological insulators. This lattice model preserves translation symmetry and six-fold rotation symmetry $C_6$ (about the hexagon center).  The single particle spectrum consists of three bands fully gapped from each other as shown in \figref{fig:Kagome}(b).

\begin{figure}[htbp]
\begin{center}
\includegraphics[width=0.9\columnwidth]{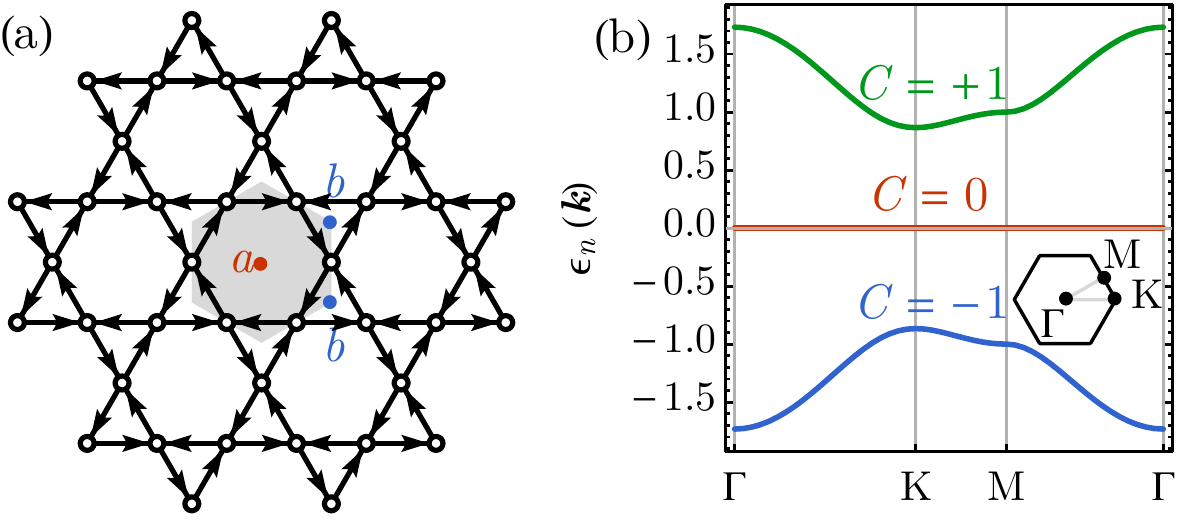}
\caption{(a) Kagome lattice model with imaginary hopping. Bound directions are specified by arrows. The gray hexagon marks out the unit cell. The Wyckoff positions $a$ and $b$ are respectively the hexagon and triangle centers. (b) The band structure of the Kagome lattice model, with the band Chern number $C$ labeled. The inset shows the Brillouin zone and high-symmetry momentum points.}
\label{fig:Kagome}
\end{center}
\end{figure}

The middle band is strictly flat with zero Chern number $C=0$, which is an obstructed trivial band. Its Wannier obstruction is only due to the lack of lattice sites at the hexagon center Wyckoff position, such that the obstruction can be lifted by adding empty sites. The top and bottom bands are Chern bands of Chern number $C=\pm1$ respectively. They combined together to form a set of fragile topological bands, which is Wannier obstructed and does not admit an effective tight-binding model description (despite their total Chern number being zero),\cite{PhysRevB.99.125122,Liu2019Shift} because the symmetry representations at high-symmetry momentum points do not match those of any atomic insulators of the same lattice symmetry. This situation is analogous to the middle bands in the tBLG around the charge neutrality, which have a fragile topological band structure for each valley.\cite{PhysRevB.98.085435,PhysRevLett.121.126402} Placing the tBLG on the aligned hexagonal boron nitride (hBN) substrate further opens up the band gap at charge neutrality. The top and bottom bands are valley Chern bands of opposite Chern numbers,\cite{2019arXiv190108110B} which resembles the Chern bands in the Kagome lattice model as in \figref{fig:Kagome}(b). In the following, we will first set aside the possible connection to tBLG systems and focus on the Kagome lattice model itself. By applying our proposed approach to analyze this toy model, we wish to gain general understanding about Wannier obstructed Mott insulators, which could facilitate future study of correlated insulating phases in Moir\'e superlattice systems.

\subsection{Nonorthogonal Localized Orbitals in Wannier Obstructed Bands}

We follow the method described in \secref{construct} to construct the localized but nonorthogonal Wannier orbitals. If we isolate the middle flat band and apply on-site repulsive interaction $U_{ij}=U_0\delta_{ij}$, the orbital $\phi_I(i)$ that minimizes the energy $H_{()}$ defined in \eqnref{eq:E} is found to be strictly localized around the hexagon as shown in \figref{fig:orbitals}(a). This orbital is similar symmetry-wise to the localized orbitals proposed in Ref.\onlinecite{PhysRevA.83.023615} for a different model, where the strict localization in both cases comes from the destructive interference between orbitals.

\begin{figure}[htbp]
\begin{center}
\includegraphics[width=0.75\columnwidth]{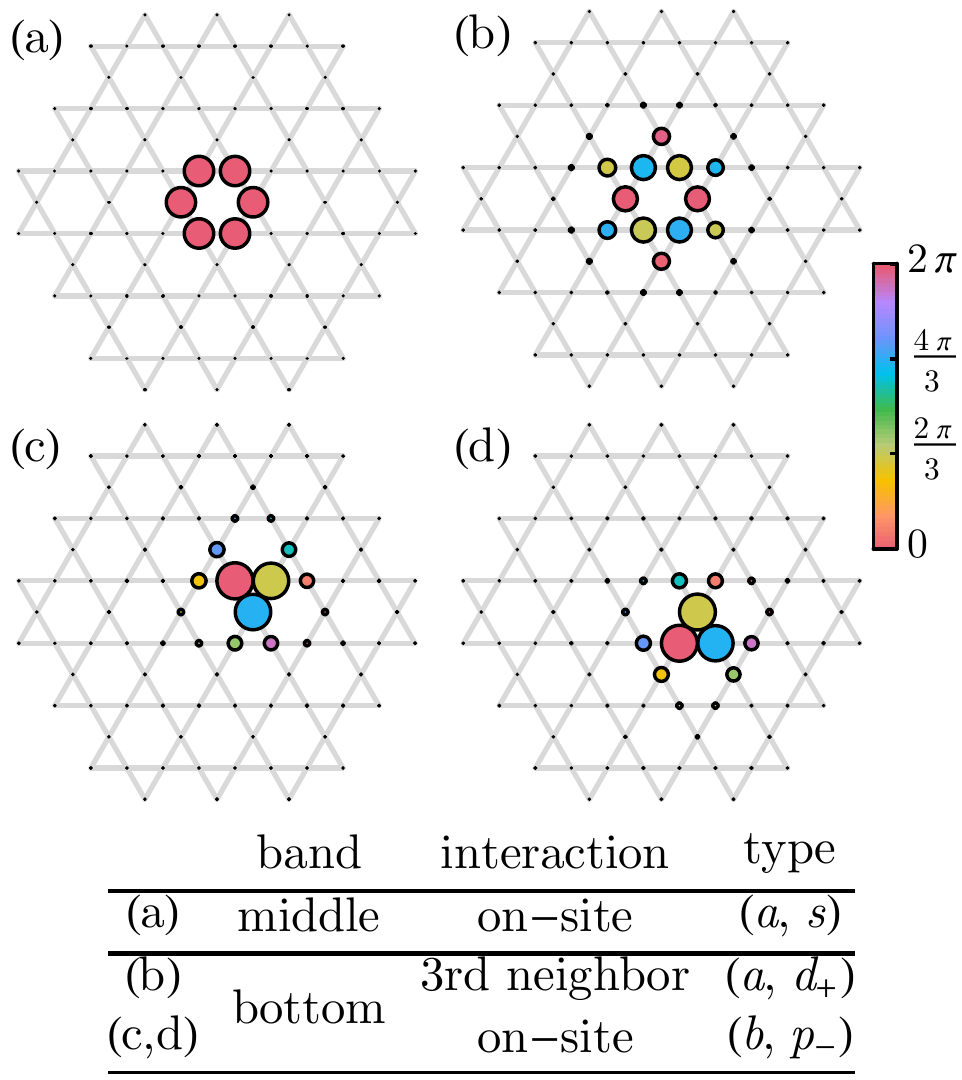}
\caption{Localized orbitals found by minimizing objective energy $H_{()}$ (the constant piece of the effective spin Hamiltonian $H_c$). The phase of the orbital wave function is specified by the colorbar.}
\label{fig:orbitals}
\end{center}
\end{figure}

If we focus on the bottom Chern band with $C=-1$ and again apply the on-site repulsion $U_0$, we found two degenerated solutions of $\phi_I(i)$ as shown in \figref{fig:orbitals}(c) and (d). These orbitals have similar features as the Wannier orbitals in twisted bilayer graphene (tBLG).\cite{PhysRevX.8.031087,PhysRevX.8.031088} They have three major peaks around the triangular lattice site and are equipped with non-zero angular momentum. To make connection to the tBLG system, we follow the notation in Ref.\onlinecite{PhysRevB.99.195455} to label these orbitals by the Wyckoff positions of their orbital centers and the angular momenta with respect to their orbital centers. The hexagon and the triangle center Wyckoff positions are denoted by $a$ and $b$ respectively as in \figref{fig:Kagome}(a), and angular momentum $0, \pm 1, \pm 2$ are denoted by $s$, $p_{\pm}$ and $d_{\pm}$. The orbital types are summarized in the table of \figref{fig:orbitals}.

\begin{figure}[htbp]

\centering
\includegraphics[width=.9\linewidth]{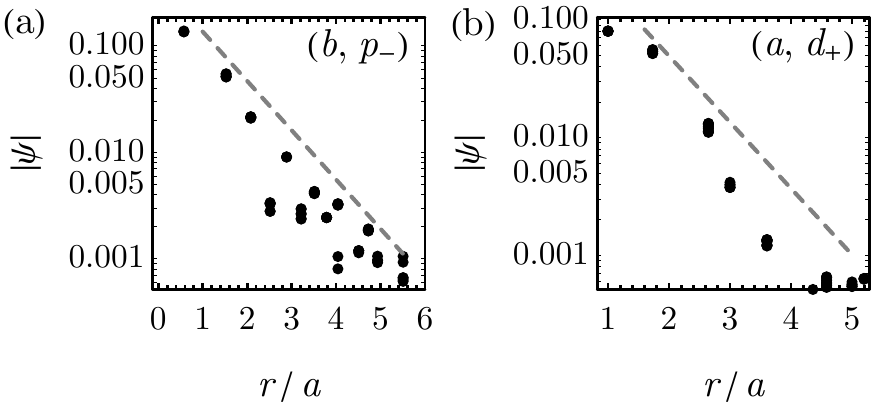}
\caption{Logarithmic norm of orbital wave functions against distance in units of bond length for (a) $(b,p_-)$ orbital  as in Fig.\ref{fig:orbitals}(c,d); (b) $(a,d_+)$ orbital  as in Fig.\ref{fig:orbitals}(b).}
\label{phase}
\end{figure}

In the Mott limit, electrons will spontaneously choose one of the $(b, p_-)$ orbitals in \figref{fig:orbitals}(c) and (d) to reside, which spontaneously breaks the $C_6$ rotation symmetry to the $C_3$ subgroup and results in the $C_3$ nematic phases. However, if we include longer range interactions, the energetically most favorable orbital can be different. For example, if we apply 3rd neighbor repulsion $U_{\la\!\la\!\la ij \ra\!\ra\!\ra}=U_3$ to the bottom Chern band, we obtain a $(a,d_+)$ orbital as shown in \figref{fig:orbitals}(b). This orbital preserves the $C_6$ rotation symmetry, so the corresponding Mott phases are not nematic. As shown in Fig.\ref{phase}, all the orbitals we constructed are indeed exponentially localized. As a result, the tensors $g_{IJ}$, $t_{IJ}$ and $U_{IJKL}$ should all decay exponentially with the inter-orbital distance, so we expect the resulting effective spin model to exhibit well controlled locality.

\subsection{Effective Spin Models and Possible Phases}

In the Mott limit, the spin dynamics for any of these orbitals can be described by the effective spin Hamiltonian on a triangular lattice, as shown in \figref{fig:triangle}. The spin-rotational symmetric Hamiltonian takes the following form
\eq{\label{eq:JK}H= J_1\sum_{\la IJ\ra}\vect{S}_I\cdot\vect{S}_J+\sum_{\la IJK \ra\in \vartriangle / \triangledown}K_{\vartriangle / \triangledown} \b{S}_I\cdot(\b{S}_J\times\b{S}_K)+...}
where three sites $IJK$ surrounding both up and down triangular plaquettes are arranged in the counterclockwise order. Again we focus on the bottom Chern band. The coupling strengths of the nearest-neighbor Heisenberg interaction $J_1$ and the chiral spin interactions $K_{\vartriangle / \triangledown}$ are plotted in \figref{JKmain} as a function of $t/U_0$ or $t/U_3$.

\begin{figure}[htbp]
\begin{center}
\includegraphics[width=0.54\columnwidth]{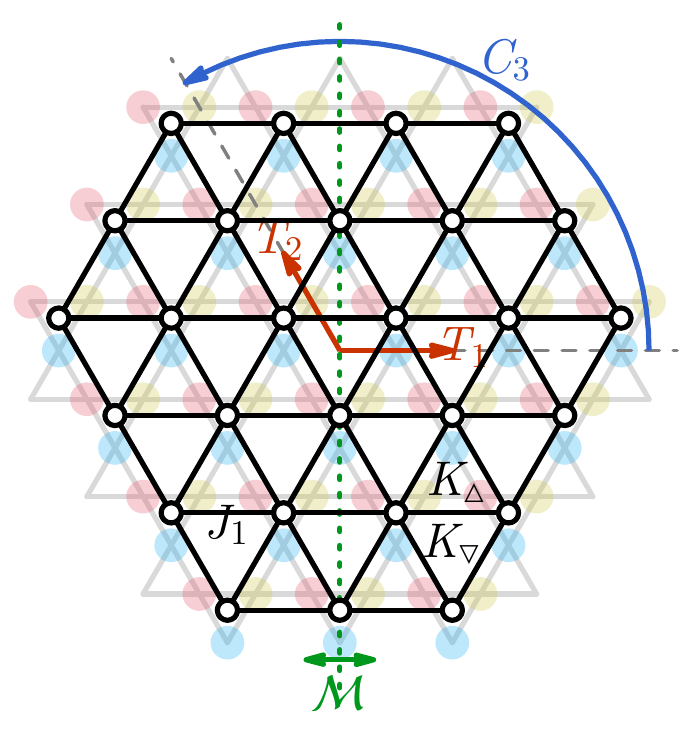}
\caption{The triangular lattice formed by the $(b,p_-)$ orbitals (in the bottom Chern band) arranged on the Kagome lattice. The effective spin model contains the Heisenberg interaction $J_1$ across the bonds and the chiral spin interaction $K_{\vartriangle/\triangledown}$ around the up/down triangles. It respects the translation $T_{1,2}$, three-fold rotation $C_3$ and anti-unitary mirror $\scM$ symmetries.}
\label{fig:triangle}
\end{center}
\end{figure}

Now we briefly discuss symmetries of this effective Hamiltonian and direct readers to \appref{app:PSG} for more details. The electron Hamiltonian has  translation symmetries $T_1$ and $T_2$ along two different directions, six-fold rotation symmetry $C_6$, and an anti-unitary mirror symmetry $\scM$. The optimal localized orbitals $\phi_I$ can spontaneously break some of the symmetries. For example, the $(b,p_-)$ orbital breaks $C_6$ to $C_3$ (see \figref{fig:triangle}). Given the symmetries of the orbitals, we can infer the symmetry of tensors $g_{IJ}$, $t_{IJ}$ and $U_{IJKL}$. It turns out that between nearest neighboring sites $I$ and $J$, $t_{IJ}$ ($g_{IJ}$) can be generated by a single parameter $t \equiv |t_{\la IJ \ra}|$ ($g \equiv |g_{\la IJ \ra}|$) given $T_{1,2}$, $C_3$ and $\scM$. Then the chiral spin interaction $K_{\triangledown}/K_{\vartriangle}$ result from $t_{IJ}$ and $g_{IJ}$ must be opposite on
neighboring triangular plaquettes, i.e. $K_{\vartriangle} = -K_{\triangledown}$. The interaction tensor $U_{IJKL}$ breaks this pattern, but $K_{\triangledown}/K_{\vartriangle}$ result from $U_{IJKL}$ is much smaller than others in this model, so we still have a good approximate symmetry $K_{\vartriangle} \simeq -K_{\triangledown}$. If we further have $C_6$ symmetry like in the case of $(a,d_+)$ orbital, $t_{IJ}$ and $g_{IJ}$ will be restricted to real numbers and furthermore the chiral spin interaction is restricted to be uniform $K_{\vartriangle} = K_{\triangledown}$. This symmetry analysis is in agreement with a previous study on a similar model.\cite{PhysRevB.99.205150}

For the $(b,p_-)$ orbital favored by the onsite interaction $U_0$, the Heisenberg interaction changes from ferromagnetic (FM) to anti-ferromagnetic (AFM) as $t$ increases (\figref{JKmain} (a)). In this case, the nonorthogonality enabled channel $t_{IJ}g_{JI}$ favoring AFM. Thus, AFM interaction starts to dominate after this new channel takes over at large $t$. The same physics happens for the $(a,d_+)$ orbital favored by the 3rd neighbor interaction $U_3$, while the transition occurs at a much smaller $t$ (\figref{JKmain} (b)). Close to the transition, there can be intermediate phases, which requires to take the chiral spin interaction into account. It turns out that near the FM-AFM transition regime, the chiral spin interaction is dominated by the $tg^2$ channel in \eqnref{eq:3spin} (see \appref{app:PSG}). The $(a,d_+)$ orbital respects the $C_6$ symmetry, which constrains the chiral spin interaction to be vanishing small (\figref{JKmain} (b)). On the other hand, the $(b,p_-)$ orbital breaks the $C_6$ symmetry, so $K_{\vartriangle}$ and $K_{\triangledown}$ are approximately related by the staggered pattern (\figref{JKmain} (a)).

\begin{figure}[thbp]
\centering
\includegraphics[width=\linewidth]{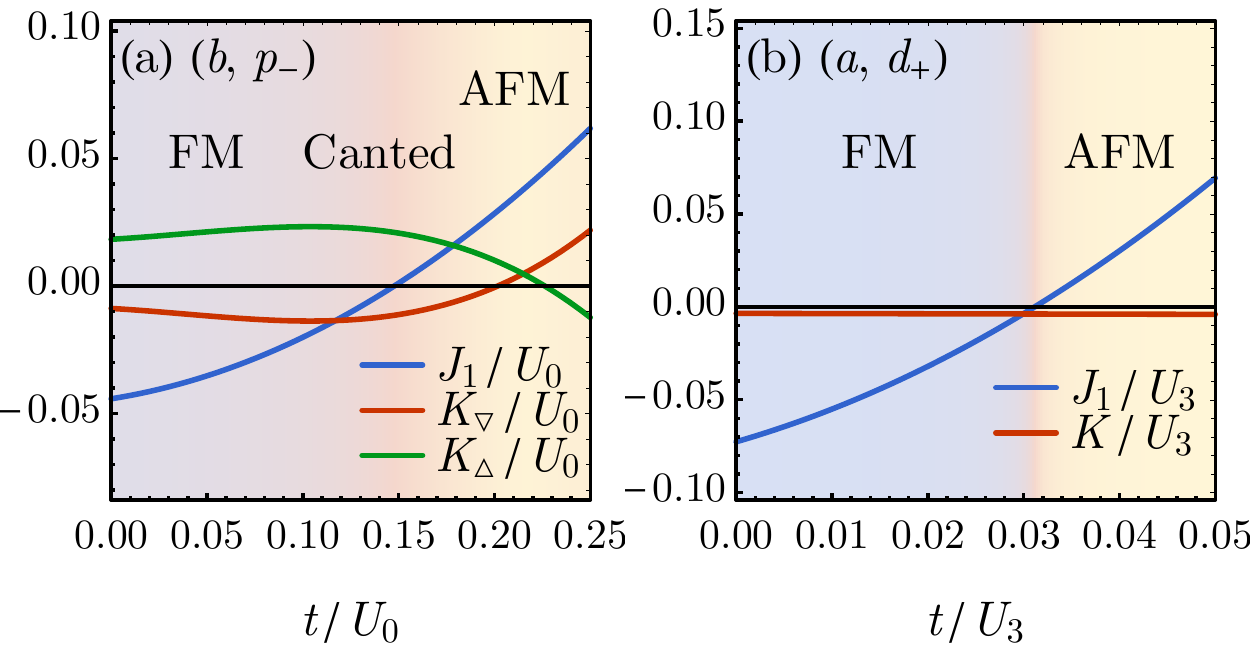}
\caption{The coupling strengths of the nearest-neighbor Heisenberg interaction $J_1$ and the chiral spin interaction $K_{\triangledown}, K_{\vartriangle}$ for (a) the $(b,p_-)$ orbital in \figref{fig:orbitals}(c) favored by the $U_0$ interaction, and for (b) the $(a,d_+)$ orbital favored by the $U_3$ interaction (where $K_{\triangledown}=K_{\vartriangle}=K$). FM/AFM stands for the (anti-)ferromagnetic phase; canted stands for the canted $120^{\circ}$ AFM configuration shown in the inset of \figref{phase_new}.}
\label{JKmain}
\end{figure}

Now we compute the coupling strengths perturbatively with two control parameters $t$ and $g$ to get insight for their behavior in more general orbitals. We approximate the orbitals by only major and secondary peaks, and the orbital wave function is determined by the angular momentum and the amplitude ratio between major and secondary peaks tunable by $g$. Most importantly, the nearest orbital overlap $g$ parameterizes the nonorthogonality of the orbitals, and controls the competition between different magnetic phases. Since $g$ only slightly depends on the magnitude of $U$, we approximate it by a constant in the following analysis. For the $(b, p_-)$ orbital, to the leading order in the hopping $t$ and the orbital overlap $g$, we have
\begin{equation} \label{eq:Jmodel}
J_1 = \frac{t^2}{3U_0} + \frac{2}{3} t g - U_0g^2,
\end{equation}
where new channels in the theory of nonorthogonal basis contribute to those terms that contain $g$. The Heisenberg coupling $J_1$ changes sign at $(U_0/t)^* = 1/g$ as shown in \figref{JKmain}(a). In the flat band limit $U_0/t \gg 1/g$, the FM exchange dominates. In \appref{app:exactFM}, we provide a rigorous proof of the ferromagnetism for Wannier obstructed Mott insulators in the flat band limit. As the band dispersion gets larger (but still on the strong coupling side) $1/g \gg U_0/t \gg 1$, the AFM interaction takes over. Due to the geometric frustration on the triangular lattice (especially when higher order spin interactions are also taken into account), several candidate orders may compete for the ground state, which we will leave for later discussion. 

From the above analysis, we expect the FM phase to become unstable toward AFM-like phases around $U_0/t\simeq1/g$. However, around this point, the exchange interaction $J_1$ tends to vanish, so we need to consider higher order interactions, e.g.~the chiral spin interaction,
\begin{equation} \label{eq:Kmodel}
K_{\triangledown}=-K_{\vartriangle}=\frac{t^3}{2U_0^2}  - \frac{t^2 g}{U_0} - \frac{5}{3} t g^2.
\end{equation}
Again, terms that contain $g$ arise from nonorthogonality enabled channels. The chiral spin interaction dominates the spin model $|K_{\vartriangle,\triangledown}|>|J_1|$ in a narrow window of $\left | U_0/t - 1/g \right | < 13/8$. A similar analysis can be repeated for the $(a,d_+)$ orbital under third neighbor repulsion $U_3$. Combining these information, we get a schematic phase diagram in \figref{phase_new}.

As a reminder, the coefficients in \eqnref{eq:Jmodel} and \eqnref{eq:Kmodel} are specific to the localized orbital in our model. However, we expect the Heisenberg interaction $J_1$ to be quadratic in $t/U$ and $g$ for generic systems, with different order-one coefficients, such that the FM-AFM cross over happens at $t/U\sim 1/g$ scale. A similar analysis can be generalized to the chiral spin interaction which is cubic in $t/U$ and $g$.

\begin{figure}[htbp]
\begin{center}
\includegraphics[width=0.74\columnwidth]{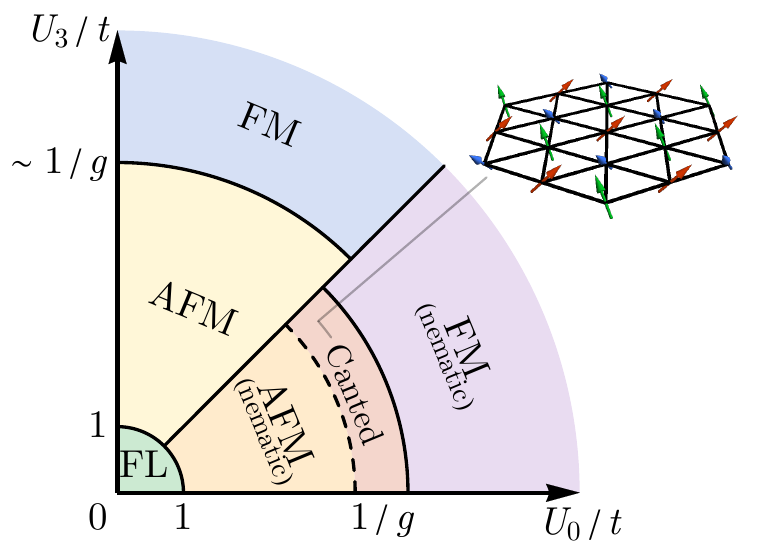}
\caption{Schematic phase diagram. The inset shows the spin configuration of the canted $120^{\circ}$ AFM order.}
\label{phase_new}
\end{center}
\end{figure}

The spin model in \eqnref{eq:JK} can give rise to different phases. We first consider the case when the on-site interaction $U_0$ dominates, which favors the $(b,p_-)$ orbital. The $(b,p_-)$ orbital breaks the $C_6$ symmetry to $C_3$ spontaneously, resulting in nematic phases. In this case, the chiral spin interaction is approximately staggered $K_{\vartriangle}\simeq -K_{\triangledown}$. We present a classical picture of possible phases in the following. Away from the $U_0/t\sim 1/g$ transition regime, the nearest Heisenberg interaction $J_1$ dominates, which leads to a Heisenberg FM state when $J_1<0$ and a $120^{\circ}$ AFM state when $J_1>0$. At the classical level, the in-plane $120^{\circ}$ AFM state can be tuned towards the $z$-axis FM state by canting the $120^{\circ}$ spin configuration in the $x$-$y$ plane toward the $z$-axis, as illustrated in the inset of \figref{phase_new}. The canted $120^{\circ}$ AFM configuration happens to have opposite spin chirality between up and down triangles, which is indeed favored by the staggered chiral spin interaction. We denote this intermediate state as the canted state in \figref{JKmain} and \figref{phase_new}, which carries both non-zero magnetization and staggered scalar spin chirality. This spin configuration has been previously studied as an umbrella-type noncoplanar phase in Ref.\onlinecite{Akagi_2011} within a Kondo system under different settings. Since both the canted state and the $120^{\circ}$ AFM state spontaneously break the spin $\mathrm{U}(1)$ symmetry, they actually belong to the same phase. In contrast, there has to be a phase transition between the canted AFM phase and the Heisenberg FM phase which preserves the spin $\mathrm{U}(1)$ symmetry (\figref{phase_new}). However, this classical picture might be modified under quantum fluctuations and longer-range geometric frustrations. We then comment on the other case when the longer-range interaction becomes important. For example, the third-neighboring interaction $U_3$ favors the $(a,d_+)$ orbital, which preserves all the lattice symmetries. In this case, the chiral spin term is parametrically small. The sign change of $J_1$ still happens when $U_3/t$ is around the order of $1/g$, which drives the transition between FM and AFM phases. Such transition is likely first order (\figref{phase_new}). Finally, when $U/t\sim 1$, the system is no longer captured by the strong coupling theory presented in this work. We simply denote the weak coupling phase as the Fermi liquid (FL) phase, whose instability should be further analyzed using weak coupling approaches. 

Let us further remark on some previous studies on the spin model in \eqnref{eq:JK} regarding more exotic phases due to quantum fluctuation. Though it is well believed that uniform flux $K_{\vartriangle} \simeq K_{\triangledown}$ can drive the system toward a chiral spin liquid (CSL) phase \cite{PhysRevB.95.035141, PhysRevB.96.075116}, it is not clear what happens in the staggered limit $K_{\vartriangle} \simeq -K_{\triangledown}$. Some numerical study\cite{bauer2013gapped} and parton construction\cite{PhysRevB.83.245131} on related models suggest a possibility of gapless spin liquid. When we further add a next-nearest-neighbor AFM coupling $J_2$ to the model \cite{PhysRevB.92.140403, PhysRevB.92.041105, PhysRevB.93.144411, PhysRevLett.123.207203}, it can drive the system toward a Dirac spin liquid (DSL) phase before the system fully develops a stripe order. These are all possible phases of the effective spin model we construct through nonorthogonal projection and perturbation. We will leave these rich possibilities for future numerical investigations.

\section{\label{sum}Conclusion}

In this work, we present a different approach to understand Mott physics in Wannier obstructed systems including Chern insulators and fragile topological systems like twisted bilayer graphene. To get around the obstruction, we sacrifice the orthogonality of Wannier basis and develop a method to construct the trial nonorthogonal Wannier orbitals by numerically optimizing the Hartree energy of the system. In the Mott limit, we fill these trial orbitals with one electron per orbital. To study the low energy spin dynamics, we systematically project the Hamiltonian to a nonorthogonal spin basis and further study perturbative corrections. This new procedure concerning nonorthogonal Wannier basis gives rise to new channels to spin interactions. For example, at the level of direct projection, we find new channels that contribute to ferromagnetic, antiferromagnetic, or spin liquid phases. We demonstrate our approach with a toy model that carries Chern bands and fragile topological bands. In this model, new channels widen the antiferromagnetic phase and enhance chiral spin interactions that may lead to rich magnetic phases. 

Our result may shed light on the magnetisms in Moir\'e superlattice systems, which often host Wannier obstructed bands.  For example, the two middle bands near the charge neutrality in twisted bilayer graphene (tBLG) are identified to be fragile topological bands.\cite{PhysRevX.8.031089,PhysRevB.99.195455,PhysRevB.98.085435} Aligning the tBLG with hexagonal boron nitride (hBN) substrate, the fragile topological bands further develop into separate Chern bands within each valley, which is analogous to the Chern bands in our toy model. The observation of ferromagnetic hysteresis \cite{ISI:000483195200043} suggests that the three-quarter-filling insulating state in such system exhibits ferromagnetism. As the tBLG band structure only preserves the $C_3$ rotation symmetry within each valley, the scenario is similar to the $U_0$ dominated case in our toy model, where the localized $(b,p_-)$ orbital exhibits the famous fidget spinner structure.\cite{PhysRevX.8.031089,PhysRevX.8.031088,PhysRevX.8.031087} 

Our analysis shows a non-vanishing inter-site ferromagnetic coupling from the Fock term due to finite overlap $g$ between nonorthognal orbitals even in the limit when the band width $W$ approaches zero. This ferromagnetism becomes unstable when the band width increases up to $Ug$ set by the interaction strength $U$ and the orbital nonorthogonality $g$, which is in consistent with previous studies on narrow Chern bands.\cite{senthil2019narrow} When the system is close to the ferromagnet-antiferromagnet crossover, our approach provides a systematic framework to write down an explicit effective spin model that enables further numerical investigation of intermediate phases. Meanwhile, new channels from our framework also open up the possibility of new and richer magnetic phases close to the crossover.

\begin{acknowledgments}
We acknowledge the stimulating discussion with Ashvin Vishwanath, Eslam Khalaf, Andreas Mielke, Cenke Xu, Senthil Todadri, Michael P. Zaletel, Siddharth Parameswaran, and Da-Chuan Lu. HYH and YZY are supported by a startup fund from UCSD. SL is supported by Ashvin Vishwanath by an Ultra-Quantum Matter grant from the Simons Foundation (651440, AV) and a Simons Investigator grant. 
\end{acknowledgments}

\bibliography{main.bib}

\clearpage
\onecolumngrid
\appendix
\section{Derivation of Perturbation Theory}\label{app:perturbation}
Suppose the Hilbert space can be split into degenerate subspaces labeled by the principle quantum number $n$. Basis states $\ket{n\alpha}$ within the subspace are labeled by the secondary quantum number $\alpha$. Assuming different subspaces are orthogonal to each other, but different basis states within each subspace can be nonorthogonal,
\eq{\braket{m\alpha}{n\beta}=\delta_{mn}G_{n\alpha\beta}.}
Define $G_n^{\alpha\beta}$ (the inverse metric) to be the inverse of $G_{n\alpha\beta}$.

Consider perturbing a Hamiltonian $H_0$ by the operator $V$ in the form of
\eqs{H(\lambda)&=H_0+\lambda V,\\
H_0&=\sum_{n}\ket{n\alpha}E_nG_n^{\alpha\beta}\bra{n\beta},\\
V&=\sum_{mn}\ket{m\alpha}G_m^{\alpha\alpha'}V_{m\alpha',n\beta'}G_n^{\beta'\beta}\bra{n\beta},}
where $\lambda$ is a small parameter controlling the perturbative expansion. The coefficients $E_n$ and $V_{m\alpha,n\beta}$ are given by
\eqs{E_nG_{n\alpha\beta}&=\bra{n\alpha}H_0\ket{n\beta},\\
V_{m\alpha,n\beta}&=\bra{m\alpha}V\ket{n\beta},}
where we have assumed that all states within the same subspace are degenerated in energy under $H_0$, i.e. $H_0\ket{n\alpha}=E_n\ket{n\alpha}$.

Under the perturbation, the degeneracy in teach subspace could be lifted. The goal is to find a new set of basis which block diagonalized the perturbed Hamiltonian $H(\lambda)$, such that
\eq{\label{eq:Heigen}H(\lambda)\ket{n\beta(\lambda)}=\ket{n\alpha(\lambda)}G_n^{\alpha\alpha'}E_{n\alpha'\beta}(\lambda).}
We can always fix the gauge such that the metric $G_n^{\alpha\alpha'}$ is invariant as we move along $\lambda$ (i.e. the gauge connection is trivial). The perturbation theory provides us a systematic method to calculate $E_{n\alpha\beta}(\lambda)$ order by order as Taylor series
\eq{\label{eq:Taylor}E_{n\alpha\beta}(\lambda)=E_nG_{n\alpha\beta}+\lambda\partial_\lambda E_{n\alpha\beta}+\frac{\lambda^2}{2}\partial_\lambda^2E_{n\alpha\beta}+\cdots,}
where we have used the fact that $E_{n\alpha\beta}(0)=E_nG_{n\alpha\beta}$ in the unperturbed limit. To evaluate the derivatives, let us first derive the  Hellmann-Feynman theorem.

We start by applying $\partial_\lambda$ to both sides of \eqnref{eq:Heigen},
\eqs{&\partial_\lambda H\ket{n\beta}+H\ket{\partial_\lambda n\beta}\\
=&\ket{\partial_\lambda n\alpha}G_n^{\alpha\alpha'}E_nG_{n\alpha'\beta}+\ket{n\alpha}G_n^{\alpha\alpha'}\partial_\lambda E_{n\alpha'\beta}\\
=&\ket{\partial_\lambda n\beta}E_n+\ket{n\alpha}G_n^{\alpha\alpha'}\partial_\lambda E_{n\alpha'\beta},}
where in the second step we have used $G_n^{\alpha\alpha'}G_{n\alpha'\beta}=\delta^\alpha_\beta$. Now overlap with $\bra{m\gamma}$ on both sides, also we have
\eq{\label{eq:withmg}\bra{m\gamma}\partial_\lambda H\ket{n\beta}+\bra{m\gamma}H\ket{\partial_\lambda n\beta}=\braket{m\gamma}{\partial_\lambda n\beta}E_n+\braket{m\gamma}{n\alpha}G_n^{\alpha\alpha'}\partial_\lambda E_{n\alpha'\beta}.}
\eqnref{eq:Heigen} implies $\bra{m\gamma}H=E_m\bra{m\gamma}$ at $\lambda=0$. Moreover,  $\braket{m\gamma}{n\alpha}G_n^{\alpha\alpha'}=\delta_{mn}G_{n\gamma\alpha}G_n^{\alpha\alpha'}=\delta_{mn}\delta_\gamma^{\alpha'}$, thus \eqnref{eq:withmg} becomes
\eq{\label{eq:FH}\bra{m\gamma}\partial_\lambda H\ket{n\beta}=\braket{m\gamma}{\partial_\lambda n\beta}(E_n-E_m)+\delta_{mn}\partial_\lambda E_{n\gamma\beta}.}
When $m=n$, \eqnref{eq:FH} implies the first Hellmann-Feynman theorem
\eq{\label{eq:FH1}\partial_\lambda E_{n\alpha\beta}=\bra{n\alpha}\partial_\lambda H\ket{n\beta}=V_{n\alpha,n\beta}.}
When $m\neq n$, \eqnref{eq:FH} implies the second Hellmann-Feynman theorem
\eq{\label{eq:FH2}\braket{m\alpha}{\partial_\lambda n\beta}=\frac{\bra{m\alpha}\partial_\lambda H\ket{n\beta}}{E_n-E_m}=\frac{V_{m\alpha,n\beta}}{E_n-E_m}.}

Now applying \eqnref{eq:FH1}, we can already evaluate the first order derivative $\partial_\lambda E_{n\alpha\beta}=V_{n\alpha,n\beta}$. Take one more derivative,
\eqs{\partial_\lambda^2 E_{n\alpha\beta}&=\bra{\partial_\lambda n\alpha}V\ket{n\beta}+\bra{n\alpha}V\ket{\partial_\lambda n\beta}\\
&=\sum_{m\neq n}\braket{\partial_\lambda n\alpha}{m\gamma}G_m^{\gamma\delta}\bra{m\delta}V\ket{n\beta}+\sum_{m\neq n}\bra{n\alpha}V\ket{m\gamma}G_m^{\gamma\delta}\braket{m\delta}{\partial_\lambda n\beta},}
applying \eqnref{eq:FH2},
\eq{\partial_\lambda^2 E_{n\alpha\beta}
=2\sum_{m\neq n}\frac{\bra{n\alpha}V\ket{m\gamma}G_m^{\gamma\delta}\bra{m\delta}V\ket{n\beta}}{E_n-E_m}.}
Substitute into \eqnref{eq:Taylor}, we arrive at
\eqs{\label{eq:pert}E_{n\alpha\beta}(\lambda)&=E_nG_{n\alpha\beta}+\lambda V_{n\alpha,n\beta}+\lambda^2\sum_{m\neq n}\frac{V_{n\alpha,m\gamma}G_m^{\gamma\delta}V_{m\delta,n\beta}}{E_n-E_m}+\cdots.}
This gives the perturbative correction to the effective Hamiltonian within each block to the order of $\lambda^2$.

The perturbative correction of the state can be calculated as well. We first evaluate the derivative
\eqs{\ket{\partial_\lambda n\alpha}&=\sum_{m\neq n}\ket{m\beta}G_m^{\beta\gamma}\braket{m\gamma}{\partial_\lambda n\alpha}\\
&=\sum_{m\neq n}\ket{m\beta}G_m^{\beta\gamma}\frac{V_{m\gamma,n\alpha}}{E_n-E_m}.}
Then the state correction to the order of $\lambda$ reads
\eq{\ket{n\alpha(\lambda)}=\ket{n\alpha}+\lambda\sum_{m\neq n}\ket{m\beta}\frac{G_m^{\beta\gamma}V_{m\gamma,n\alpha}}{E_n-E_m}+\cdots.}

\section{\label{sec:convergence} Effective Spin Hamiltonian and Convergence of Permutations}
The spin Hamiltonian $\tilde{H}$ (\eqnref{eq:GH formula}) necessarily carries non-local spin interactions arising from apart permutations. At order $\mathcal{O}(g^4)$, orbital $I_1$ and $J_1$ would support a non local spin interaction $\sum_{(I_2J_2)}(\vect{S}_{I_1}\cdot\vect{S}_{J_1})(\vect{S}_{I_2}\cdot\vect{S}_{J_2})$, where the summation is taken over the entire lattice. Thus, as long as the overlapping weight $g$ is non-zero, the single site energy in FM phase would blow up after enumerating over infinite number of lattice sites. This contrasts to the fact that the Hubbard model on the Kagome lattice is well-defined in the thermodynamical limit. The bottom line is that the many-body overlapping matrix $G$ also contains non-local spin interaction, which eventually cancels out those terms in $\tilde{H}$. Specifically, we can factor out $G$ from $\tilde{H}$ by adding residue terms with colliding indices

\begin{equation}
    \label{eq:series}
    \begin{aligned}
        \tilde{H}&=G\times\bigg\{\sum_{\mcP_0\in S^*_N}(-)^{\mcP_0} \rchi_{\mcP_0}H_{\mcP_0}-\sum_{\substack{\mcP_0,\mcP_1\in S^*_N\\ \mcP_0 \cap \mcP_1\neq\emptyset}}(-)^{\mcP_0+\mcP_1} \rchi_{\mcP_0\circ\mcP_1}H_{\mcP_0}G_{\mcP_1}+...\\
        &\phantom{=}+(-)^n \sum_{\substack{\mcP_i\in S^*_N\\ \mcP_i \cap \{\mcP_j\}/\mcP_i\neq\emptyset}}(-)^{\sum_{i=0}^n\mcP_i} \rchi_{\mcP_0\circ\mcP_1...\circ\mcP_n}H_{\mcP_0}\prod_{i=1}^nG_{\mcP_i}+...\bigg\}.
    \end{aligned}
\end{equation}
In the main text, we show an diagram representation of it in Eq.\ref{eq:H=Gexpand}, and we identify the first term as $H_c$. However, the oscillating series still contain infinite terms, and the convergence of the entire series is not gauranteed. To resolve these puzzles, we perform a numerical test on a finite-size system, and then give a general argument on the convergence of entire series.

First, we notice that for $\mcP'_0=\mcP_0\circ\mcP_1$, the first two terms in Eq.(\ref{eq:series}) have the same spin operator $\rchi_{\mcP'_0}=\rchi_{\mcP_0\circ\mcP_1}$ but with different strength $H_{\mcP'_0}>H_{\mcP_0}G_{\mcP_1}$, because the second formula necessary contains redundant $t,g,U$ at the intersection $\mcP_0 \cap \mcP_1\neq\emptyset$. Therefore, we would like to claim that $H_c$ dominates over others. To this end, we perform a numerical test of it on a finite-size system. We define the deviation between $G^{-1}\tilde{H}$ and $H_c$ as
\be
D=1-\frac{||G H_c||}{||\tilde{H}||}.
\ee
\begin{figure}[htbp]
 \includegraphics[width=.8\linewidth]{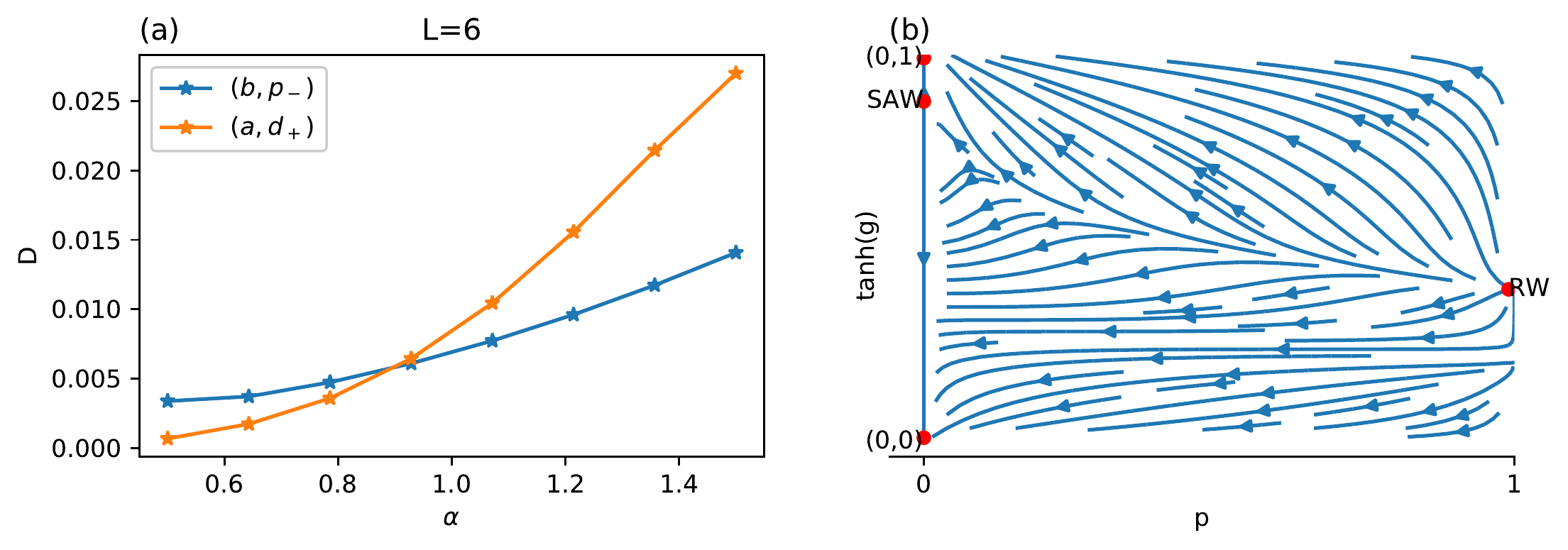}
\caption{In (a), we plot $D$ against rescaled weight $g'=\alpha g$ for two different nonorthogonal Wannier orbitals (Fig.\ref{fig:orbitals}) on a six-site system. For $\alpha=1$ we have the original configuration. In (b), we plot RG flow of Eq.(\ref{rgflow}), in which red points indicate fixed points.}
\label{D}
\end{figure}
The relation between $D$ and typical overlapping weight $g$ on a finite-size triangular lattice is plotted in Fig.\ref{D}. We artificially rescale the weight of each orbital $g'=\alpha g$ , which is equivalent to transferring weights onto the Wyckoff position $a$ and $b$. They do not belong to the kagome lattice thus have no contribution on permutations. We find that $D(\alpha=1)<0.01$, meaning a high accuracy of approximating $G^{-1}\tilde{H}$ with $H_c$. In conclusion, within numeric capability, $H_c$ is a promising starting point to investigate the physics of the model.

Next, we want to address the issue related to the convergence of the whole series. We take the FM phase to examine this convergence. The energy of FM phase $\la H_c\ra_\text{FM}$ is calculated using  $\la \rchi \ra_\text{FM}=1$.
Generically, there will be a $g^l$ total weight in front of each term with $l=|\mcP|$ being the length of the connected permutation. When $l$ increases, the number of graphs also grows, but we have little knowledge of the speed. We need to determine which factor will dominate in the thermodynamical limit. To this end, we first transform the energy into the language of random walk. Since each term only encounters one $H_\mcP$, and each individual $H_\mcP$ only differs from $G_\mcP$ by a local factor $t_{\mcP(I)I}$ or $U_{\mcP(I)I\mcP(J)J}$, we are allowed to regard them as $g_{IJ}$ correspondingly without changing the convergence of the series. Then we have
\be
\beal
G_{\mcP\in S^*_N}&\to g^{|\mcP|}\la0|\text{1-loop}(|\mcP|)|0\ra\\
H_{\mcP\in S^*_N}&\to |\mcP|g^{|\mcP|}\la0|\text{1-loop}(|\mcP|)|0\ra
\eeal
\ee
where $|.|$represent the length of connected permutation and the factor $|\mcP|$ for $H$ comes from the fact that we can choose any bond to be $t_{\mcP(I)I}$. The $\la0|W|0\ra$ means the total number of graphs starting from origin and ending at it under rule $W$. And $W=\text{1-loop}(|\mcP|)$ means we can only take a closed connected permutations with no self-intersections. 
Applying $(-)^n(-)^{\sum_{i=0}^n\mcP_i}=(-)^n(-)^{\sum_{i=0}^n(|\mcP_i|-1)}=(-)^{-1+\sum_{i=0}^n|\mcP_i|}$, we have
\be
\la H_c\ra_\text{FM}=\frac{N}{2}\sum_l l(-g)^l \la0|\text{RW}(l)|0\ra \cdot \prod_{i=1}^{n_c}\frac{1}{k_i}
\ee
where $k_i$ is the overcounting factor for each intersection (e.g. $k_i=2$ for site $i$ being visited twice). This is because whenever there is an intersection, simple random walk (RW) will have multiple choice to go through, which overcounts the $\la H_c\ra_\text{FM}$. In fact, for backtracking process $\leftrightarrows$, there should be no discounting factor since there is only one way for RW to act. But neglecting backtracking does not influence the results of RW much, especially when the number of neighbors is large. Generally, the random walk problem can be written as
\be
f(g, p)=\frac{N}{2}\sum_l l(-g)^l p^{n_c}\la0|\text{RW}(l)|0\ra, 
\label{gmp}
\ee
where $0<g, p<1$. Notice, when $p=1$, we recover the RW problem, while when $p=0$, it becomes self-avoiding walk (SAW). Following \cite{kardar_2007}, we decompose number of walks with length $l$ into $\la0|W(l)|0\ra=\la0|T^l|0\ra$ with transfer matrix $T=W(1)$. The matrix elements of transfer matrix in position space can be written down explicitly. 
Applying Fourier transformation, the transfer matrix becomes diagonalizable due to the translation symmetry. The energy is calculated as
\be
\beal
\frac{E}{N}&=\frac{1}{2}\sum_l l(-g)^l \la 0 | W(l)|0\ra\\
&=\frac{1}{2}\Tr\bigg(\sum_l l(-gT)^l \bigg)\\
&=\frac{1}{2}\Tr[ \log(1+gT)]\\
&=\frac{1}{2}\sum_{\b q}\Tr[ \log(1+gT(\b{q}))]
\eeal
\ee
For RW and SAW, we have
\be
\beal
\frac{E_\text{RW}}{N}&=\frac{1}{2} \int \frac{d^2\b{q}}{(2 \pi)^2}\log\bigg\{1+2g(\cos{q_x}+\cos{q_y}+\cos{(q_x-q_y)})\bigg\}\\
\frac{E_\text{SAW}}{N}&=\frac{1}{2}\int\frac{d^2\b q}{(2\pi)^2} \log\bigg\{1+g^2(3-8g+3g^2+g^4)+2g(1-g^2)^2(\cos{q_x}+\cos{q_y}+\cos{(q_x-q_y)})\bigg\},
\eeal
\ee
respectively. In both cases, the critical point is given by setting $q_x=q_y=\pi$. Specifically, in RW, the critical point is $g^*_\text{RW}=1/2$, larger than which the energy keeps diverging; in SAW, there is only one zero points at $g^*_\text{SAW}=1$, apart from which, the energy is finite. Then one would wonder what is the fate of Eq.(\ref{gmp}) when $0<p<1$ --- Does it diverge like RW when $g>g^*$ or like SAW when $g=g^*$? We apply the real space renormalization group \cite{ISI:A1984SP43400052} (RG) approach to investigate the critical point of random walk. F. Family \textit{et al.} argued that for bond dimension $b$, random walk of length $\xi=(d-1)(b-1)^2+b^2$ is enough to capture the critical point. Consequently, we obtain the recursion relation for $b=2$ on a triangular lattice,
\be
\beal
g' =&-g^2+5g^3-g^4(8+4p+2p^2)+g^5(4+32p+23p^2)\\
g'^2p'=&g^4p^2-g^5(4p+6p^2)+g^6(4p+28p^2+17p^3+4p^4)-g^7(48p^2+158p^3+52p^4)\\
&+g^8(8p^2+236p^3+448p^4+88p^5+6p^6)-g^9(90p^3+696p^4+746p^5+120p^6)\\
&+g^{10}(338p^4+1607p^5+1376p^6+160p^7)
\eeal
\label{rgflow}
\ee
There are two stable fixed points $(g=\infty,p=0)$ and $(g=0,p=0)$ and four unstable fixed points $(g\approx1.25,p=0)$, $(g=\infty,p=1)$, $(g\approx0.4,p=1)$ and $(g=0,p=1)$. For fixed point $(g\approx1.25,p=0)$, it is irrelevant in $p$ but relevant in $g$. We plot the RG flow in Fig.\ref{D}(b), and find for any $0<p<1$, the RW belongs to the same universality class of SAW, implying only divergence at a critical $g^*(p)$ . Thus the FM phase in our problem is indeed well-defined except at a point. In our model (\secref{model}), the $g$ is even below the critical point $g*_\text{RW}=1/2$ of RW.

\section{Exponentially Localized Orbitals for a Chern Band}\label{ChernBandWannier}
Consider an isolated band characterized by the Bloch states $\ket{\psi_{\bs k}}$, which are normalized as\footnote{Here we use a normalization which is convenient for taking the infinite size limit. } $\inner{\psi_{\bs k}}{\psi_{\bs k'}}=N\delta_{\bs k,\bs k'}=(2\pi)^d\delta(\bs k-\bs k')$, where $N$ is the number of unit cells and $d$ is the spatial dimension. If $\ket{\psi_{\bs k}}$ is smooth in the Brillouin zone (BZ) torus (which implies periodicity), the conventional Wannier orbitals \cite{RevModPhys.84.1419} defined as 
\begin{align}
	\ket{\phi_{\bs R}}=\frac{1}{N}\sum_{\bs k}\rme^{-\rmi \bs k\cdot \bs R}\ket{\psi_{\bs k}}, 
\end{align}
with $\bs R$ being Bravais lattice vectors labeling the unit cells, are exponentially localized, orthonormal, and related to each other by translations. However, when $d=2$ and the band has a nonzero Chern number, such a smooth gauge is never possible, and one can not find $N$ number of orbitals satisfying the three properties simultaneously. 
If we can find an unnormalized smooth gauge $\ket{\psi'_{\bs k}}=\lambda_{\bs k}\ket{\psi_{\bs k}}$ and define $\ket{\phi'_{\bs R}}$ orbitals in the same way as above, then these orbitals will not be orthonormal, but are still exponentially localized and related to each other by translations. Although this sounds like a good deal, there is an important issue here: $\ket{\psi'_{\bs k}}$ has to vanish at some point in the BZ, otherwise we can normalize it and obtain a smooth normalized gauge. This implies that the orbitals $\ket{\phi'_{\bs R}}$ do not form a complete basis of the subspace. 
In the following, we will show that, in the simplest situation where the smooth gauge $\ket{\psi'_{\bs k}}$ vanishes at a single point, which is indeed true for the Kagome lattice model we considered in this work, one can find a set of \emph{complete and exponentially localized} orbitals for a Chern band with both orthonormality and the translation symmetry sacrificed. 

Suppose the smooth gauge $\ket{\psi'_{\bs k}}$ vanishes at $\bs k_c\in {\rm BZ}$ and is nonvanishing elsewhere. We can find another smooth gauge $\ket{\psi''_{\bs k}}$ which is nonzero at $\bs k_c$, and define another set of localized orbitals $\ket{\phi''_{\bs R}}$. Let $\bs R_c$ be some arbitrary unit cell, we now prove that 
\begin{align}
	\{\ket{\phi'_{\bs R}}|\bs R\neq \bs R_c\} \cup \{\ket{\phi''_{\bs R_c}}\}
\end{align}
is a complete basis. Note that these orbitals preserve the translation symmetry except for a single unit cell $\bs R_c$. Since $\ket{\phi''_{\bs R_c}}$ contains a momentum component which is absent in all $\ket{\phi'_{\bs R}}$, it suffices to prove that $\ket{\phi'_{\bs R}}$ with $\bs R\neq \bs R_c$ are linearly independent, which follows from the following lemma. 
\begin{lemma}
Let $M_{\bs k,\bs R}=\rme^{-\rmi\bs k\cdot\bs R}$ with $\bs k\neq \bs k_c$ and $\bs R\neq \bs R_c$ be the matrix of an incomplete Fourier transform. $M$ is invertible with the inverse explicitly given by
\begin{align}
   \left(M^{-1}\right)_{\bs R,\bs k}=\frac{1}{N}\rme^{\rmi \bs k\cdot \bs R}\left( 1-\rme^{-\rmi (\bs k-\bs k_c)\cdot(\bs R-\bs R_c)} \right). 
\end{align}
\end{lemma}

Next, we shall discuss whether the dual orbital formalism is applicable in the case of a Chern band. Denote the exponentially localized orbitals we constructed above by $\{\ket{\varphi_{\bs R}}\}$, the dual orbitals $\{\ket{\tilde\varphi_{\bs R}}\}$ are defined by $\inner{\tilde\varphi_{\bs R}}{\varphi_{\bs R'}}=\delta_{\bs R,\bs R'}$. It is clear that $\ket{\tilde\varphi_{\bs R_c}}$ should just be proportional to $\ket{\psi''_{\bs k_c}}$ and is delocalized. However, there is no need to care about a single orbital, and it is more important to check whether other $\ket{\tilde\varphi_{\bs R}}$ are exponentially localized or not. 
Let $\{\ket{\tilde\phi_{\bs R}'}|\bs R\neq \bs R_c\}$ be the dual basis to $\{\ket{\phi_{\bs R}'}|\bs R\neq \bs R_c\}$ in the subspace with $\bs k\neq \bs k_c$, then we have
\begin{align}
    \ket{\tilde\varphi_{\bs R}}=\ket{\tilde\phi_{\bs R}'}+\lambda_{\bs R}\ket{\psi''_{\bs k_c}}~~~(\bs R\neq\bs R_c). 
\end{align}
The values of $\lambda_{\bs R}$ are not important; adjusting these will not affect the property $\inner{\tilde\varphi_{\bs R}}{\varphi_{\bs R'}}=\delta_{\bs R,\bs R'}$ for $\bs R,\bs R'\neq \bs R_c$. Since $\ket{\phi'_{\bs R}}=\frac{1}{N}\sum_{\bs k\neq \bs k_c}\ket{\psi'_{\bs k}}M_{\bs k,\bs R}$, we have 
\begin{align}
    \ket{\tilde\phi_{\bs R}'}=\sum_{\bs k\neq \bs k_c}\ket{\tilde\psi'_{\bs k}}\left(M^{-1}\right)^*_{\bs R,\bs k}, 
\end{align}
where $\ket{\tilde\psi'_{\bs k}}\propto\ket{\psi'_{\bs k'}}$ and is normalized as $\inner{\tilde\psi'_{\bs k}}{\psi'_{\bs k}}=N$. Given some standard Bloch basis, e.g. labeled by sublattice indices, the Bloch states can be represented as $u$-vectors. Let $\ket{\psi'_{\bs k}}$ be represented by $u'_{\bs k}$ which takes the following form when $\delta\bs k:=\bs k-\bs k_c\approx 0$: 
\begin{align}
    u'_{\bs k,\alpha}=\bs v'_\alpha\cdot\delta\bs k+\mc{O}(\delta\bs k^2),  
\end{align}
where $\alpha$ labels the $u$-vector components. Then $\ket{\tilde\psi'_{\bs k}}$ are represented by 
\begin{align}
    \tilde u'_{\bs k,\alpha}=\frac{\bs v'_\alpha\cdot\delta\bs k}{\sum_\beta |\bs v'_\beta\cdot\delta\bs k|^2}. 
\end{align}
Using the expression for $M^{-1}$, we find that the Fourier transform of $\ket{\tilde\phi_{\bs R}'}$ is proportional to
\begin{align}
    \frac{[(\bs R-\bs R_c)\cdot\delta\bs k](\bs v'_\alpha\cdot\delta\bs k)}{\sum_\beta |\bs v'_\beta\cdot\delta\bs k|^2}~~~(\delta\bs k\neq 0)
\end{align}
in the $u$-vector representation, which is finite but not continuous as $\delta\bs k\rightarrow 0$. Therefore $\ket{\tilde\varphi_{\bs R}}$ with $\bs R\neq\bs R_c$ can not be made exponentially localized, implying that the dual orbital formalism is not a good approach for the orbitals we constructed. 

\section{\label{app:PSG}Symmetry Analysis of Chiral Spin Interactions}

The Kagome lattice model considered in this work has the following lattice symmetries: translation symmetries $T_1$ and $T_2$ along two different directions, six-fold rotation symmetry $C_6$, and an anti-unitary mirror symmetry $\scM$ (reflection along the vertical axis and followed by complex conjugation). They are illustrated in \figref{fig:PSG}(a).

Depending on the optimal localized orbital achieved by minimizing the energy $H_{()}$, part of the lattice symmetry can be spontaneously broken. For example, the $(b,p_-)$ orbital in \figref{fig:orbitals} breaks the six-fold rotation symmetry $C_6$ to three-fold $C_3$, because there are two Wyckoff positions $b$ in each unit cell and choosing one of them to occupy will necessary break the lattice symmetry. On the other hand the $(a,d_+)$ orbital respects the $C_6$ symmetry, as it transforms under the $C_6$ symmetry as an irreducible representation. Under symmetry action, the localized orbital $\phi_I$ transforms as
\eqs{T_{1,2}:&\phi_I\to\phi_{T_{1,2}(I)},\\
 C_3:&\phi_I\to e^{-\ii 2\pi/3}\phi_{C_{3}(I)},\\
  \scM:&\phi_I\to\phi^*_{\scM(I)}.}
If the orbital further respect the $C_6$ symmetry, we also have \eq{C_6:\phi_I\to e^{\ii 2\pi/3}\phi_{C_{6}(I)}.}
Here $G(I)$ denotes the new orbital index that $I$ transforms to under the symmetry group element $G$. To preserve the translation symmetry, the orbitals follows the arrangement of the unit cells, and form a triangular lattice. The orbital (or unit cell) index $I$ can be considered as the site index on the triangular lattice. In the Mott state, the electron spin degrees of freedom will ret on these sites.

\begin{figure}[htbp]
\begin{center}
\includegraphics[width=0.6\columnwidth]{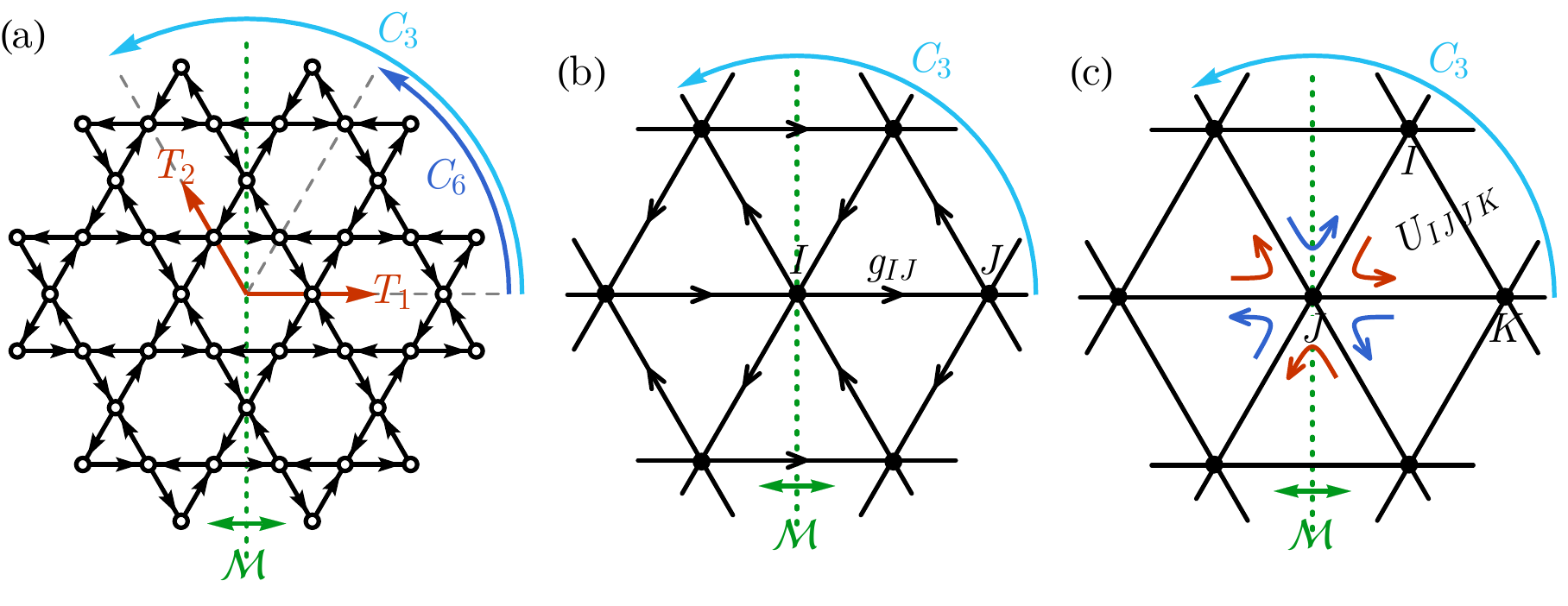}
\caption{(a) Symmetries of the lattice model. Consider $(b,p_-)$ orbitals that breaks the $C_6$ symmetry, the remaining symmetries fix the pattern of (b) $g_{IJ}$ (same as $t_{IJ}$) and (c) $U_{IJJK}$.}
\label{fig:PSG}
\end{center}
\end{figure}

Given the symmetry transformations of the orbitals $\phi_I$, we can infer the symmetry transformations of the tensors $g_{IJ}$, $t_{IJ}$ and $U_{IJKL}$, which are defined via
\eqs{g_{IJ}&=\sum_{i}\phi^*_I(i)\phi_J(i),\\
t_{IJ}&=\sum_{ij}\phi^*_I(i)t_{ij}\phi_J(j),\\
U_{IJKL}&=\sum_{ij}\phi^*_I(i)\phi_J(i)U_{ij}\phi^*_K(j)\phi_L(j).\\}
For unitary symmetries $G=T_1,T_2,C_3,C_6$, they transform as
\eq{G:g_{IJ}\to g_{G(I)G(J)},\quad t_{IJ}\to t_{G(I)G(J)},\quad U_{IJKL}\to U_{G(I)G(J)G(K)G(L)},}
such that tensor elements related by the symmetry should simply be equal to each other. Only anti-unitary symmetry $\scM$ relates them by additional complex conjugate
\eq{\scM:g_{IJ}\to g^*_{\scM(I)\scM(J)},\quad t_{IJ}\to t^*_{\scM(I)\scM(J)},\quad U_{IJKL}\to U^*_{\scM(I)\scM(J)\scM(K)\scM(L)},}
Note that $t_{IJ}$ has identical symmetry property as $g_{IJ}$. We can use symmetry transformations to bring one tensor element at a particular link or plaquette to elsewhere through out the triangular lattice. For example, between the nearest neighboring sites $I, J$, required by the symmetries $T_1,T_2,C_3,\scM$,
\eq{g_{IJ}=g^*_{JI}=g,\quad t_{IJ}=t^*_{JI}=t,}
if the direction $I\to J$ follows the link directions as depicted in \figref{fig:PSG}(b). If we further require $C_6$ symmetry, parameters $g$ and $t$ will be restricted to real numbers. Among three sites $I,J,K$ in an upper/lower-triangle following the counterclockwise order, the symmetries $T_1,T_2,C_3,\scM$ requires all $U_{IJJK}$ terms to be related as
\eq{U_{IJJK}=U^*_{KJJI}=U_{\vartriangle/\triangledown},}
see \figref{fig:PSG}(c). In the absence of the $C_6$ symmetry, $U_{\vartriangle}$ and $U_{\triangledown}$ are not related in general. If we impose the $C_6$ symmetry, we have $U_{\vartriangle}=U_{\triangledown}$, but they are still in general complex.

Based on \eqnref{eq:3spin}, the spin chirality term mainly originates from two channels
\begin{equation}
K_{IJK} = 4 \Im(t_{IJ}g_{JK}g_{KI}+U_{IJJK}g_{KI}+\text{perm.}).
\end{equation}
In terms of the parameters $t,g,U_{\vartriangle/\triangledown}$, we found
\eqs{K_{\vartriangle}&=4\Im(t g^2+U_{\vartriangle} g)=4\Im(t g^2)+4\Im(U_{\vartriangle} g),\\ K_{\triangledown}&=4\Im((t g^2)^*+U_{\triangledown} g^*)=-4\Im(t g^2)+4\Im(U_{\triangledown} g^*).}
For $(b,p_-)$ orbitals that does not have the $C_6$ symmetry, $K_\vartriangle$ and $K_\triangledown$ are not related in general. If $\Im(tg^2)$ term dominates, the spin chirality term will be approximately staggered $K_\vartriangle\simeq - K_\triangledown$. For $(a,d_+)$ orbitals that respects the $C_6$ symmetry, the spin chirality term is uniform $K_\vartriangle=K_\triangledown$.

\section{\label{app:exactFM} Rigorous Statements About Ferromagnetism}
In the case of on-site Hubbard interaction, we can make some rigorous statements about ferromagnetism in the Kagome lattice model discussed in the main text. 

We first introduce an important theorem and a corollary due to Andreas Mielke about ferromagnetism in general flat-band Hubbard models\cite{ISI:A1993KU26700019}, before focusing on the specific model. 
Consider the Hubbard model of spin-1/2 electrons, defined on a finite lattice $\Lambda$ by the Hamiltonian: 
\be
H=\sum_{x,y,\sigma}t_{xy}c^{\dagger}_{x,\sigma} c_{y,\sigma}+U\sum_x n_{x,+}n_{x,-},  
\ee
where $T=(t_{xy})_{x,y\in\Lambda}$ is a Hermitian matrix and $U>0$. We assume without loss of generality that the matrix $T$ is non-negative (positive semi-definite) and has a lowest eigenvalue $0$ with multiplicity $N_d$. Let $\{ \phi_i(x),~i=1,\cdots,N_d\}$ be an orthonormal basis of the kernel of $T$. We define the corresponding fermion mode creation operators $f^\dagger_{i,\sigma}=\sum_x \phi_i(x)c^\dagger_{x,\sigma}$. Now suppose the number of electrons $N_e$ in the system is equal to $N_d$, we know the following spin-polarized state
\be
\ket{\psi}=\prod_i f^\dagger_{i,+}\ket{\mathrm{vac}}
\ee
is an exact ground state of the Hamiltonian. This does not immediately imply that the system exhibits ferromagnetism; we at least need to check whether this ground state is unique up to the spin rotation degeneracy. We introduce the two-point equal-time correlation function of the state $\ket\psi$, 
\be
C_{x,y}:=\bra\psi c^\dagger_{x,+} c_{y,+}\ket\psi. 
\ee
Let us also define the following terminology for the simplicity of discussions. 
\begin{definition}	We say a correlator matrix $(C_{x,y})$ is \underline{connected}, if one cannot use simultaneous row and column permutations to transform it into a block-diagonal form with more than one block being nonzero.
\end{definition}
\noindent In other words, $(C_{x,y})$ is connected if it is irreducible after removing vanishing rows and columns. Now we can state the main theorem.\cite{ISI:A1993KU26700019}
\begin{theorem}[Mielke]\label{MielkeThm}
	The state $\ket{\psi}$ is the unique ground state of $H$ with $N_e=N_d$ electrons up to the spin degeneracy if and only if $(C_{x,y})$ is connected. 
\end{theorem}
\noindent We would like to remark that in the original paper, $T$ is assumed to be real symmetric, but the proof of the theorem actually applies to general complex Hermitian hopping matrices.

Now we consider a more complicated situation where zero is not the lowest eigenvalue of $T$. Let $N_<$ be the number of eigenvalues of $T$ below zero. Let $\ket{\psi}$ and $C_{x,y}$ be defined in the same way as before, i.e. only the zero energy states are occupied. Mielke\cite{ISI:A1993KU26700019} also derived the following corollary using degenerate perturbation theory. 
\begin{corollary}[Mielke]\label{MielkeCor}
	Assuming translation symmetry, if $(C_{x,y})$ is connected, then for a sufficiently small $U$ (for a fixed lattice $\Lambda$), the ground state with $N_e=2N_<+N_d$ electrons is spin-polarized with total spin quantum number $S=N_d/2$, and it is unique up to the spin degeneracy. 
\end{corollary}

Equipped with the above general results, we can now study the Kagome lattice model which has the hopping Hamiltonian
\be
	H_0=\frac{1}{2}\sum_{i\leftarrow j}(i c^\dagger_i c_j+h.c.)
\ee
and is illustrated in Fig.~\ref{fig:Kagome}a. We will always consider periodic boundary condition, i.e. the system lives on a torus, so that there is a translation symmetry. Let $N$ be the total number of unit cells. We turn on a repulsive on-site Hubbard interaction: 
\be
	H=H_{ 0}+U\sum_in_{i,+}n_{i,-}. 
	\label{PoHubbardModel}
\ee
Using Corollary~\ref{MielkeCor}, we are able to prove the following result for the middle flat band. 
\begin{theorem}
	Fixing a periodic lattice with $N>3$, when $U$ is sufficiently small, the half-filling ground state of $H$ is spin-polarized with $S=N/2$, and it is unique up to the spin degeneracy. 
\end{theorem}
\begin{proof}
Due to Corollary~\ref{MielkeCor}, it suffices to check the correlator matrix $(C_{\b x,\b y})$ is connected. If it is not, we can find a nonempty proper subset $A$ of the whole lattice $\Lambda$, such that the correlation function between any site in $A$ and any site in the complement $A^c$ is zero. This implies the existence of a pair of nearest-neighbor sites whose correlator vanishes. For the middle flat band of the Kagome lattice model, we find that all nearest-neighbor correlators take the same real value $C_{nn}$ for any fixed periodic lattice. It is not hard to prove that $C_{nn}$ is nonzero whenever $N>3$ and this proves the theorem. The details are not important so we skip them here. 
In particular, if we fix the modulus parameter (shape) of the real space torus and take the infinite size limit, $C_{nn}$ converges to a momentum integral which we numerically evaluated to be around $0.11$. This already proves a slightly weaker statement where, instead of considering all possible lattices with $N>3$, we fix a modulus parameter and take $N$ large enough. 
\end{proof}
\noindent We are not able to directly say anything about the lowest band as it is not exactly flat. However, if we apply a band flattening, then the following result easily follows from Theorem~\ref{MielkeThm}. 
\begin{theorem}
Suppose we flatten the lowest band in the Kagome lattice model and turn on an on-site Hubbard interaction. Given a modulus parameter (shape) of the real space torus, when $N$ is sufficiently large, the ground state of the system with $N$ number of electrons is fully spin-polarized and is unique up to the spin degeneracy for any $U>0$. 
\end{theorem}
\begin{proof}
Due to the translation and rotation symmetries of the Kagome lattice model, all nearest-neighbor correlators $C_{\b x,\b y}$ of the lowest band with $\b y\rightarrow \b x$ (say) along the arrows in Fig.~\ref{fig:Kagome} take the same value $C_{nn}'$. If we fix the modulus parameter of the real space torus and take the infinite size limit, $C_{nn}'$ converges to a momentum integral which we numerically evaluated to be around $-0.057 + 0.20\ii$. 
\end{proof}

\section{\label{app:phase} Phase Transition and Other Wannier-Obstructed Bands}
In this section, we give a discussion about the numerical results of ground state properties of the toy model in Sec.\ref{model}. The localized non-orthogonal orbital is determined from Eqn.\ref{eq:GPeqn} without band flattening to take finite-band-width effect into account at zeroth-order approximation. Then the effective spin Hamiltonian is constructed by enumerating all the possible permutations on a finite-size lattice. We use exact-diagonalization to investigate different order parameters in Fig.\ref{bottpt}: $M_\text{FM}=\sum_I\la Z_I\ra/2\tilde{N}$ and $CSO=\la \b{S}_0\cdot(\b{S}_1\times\b{S}_2) \ra$, where $Z_I$ is the on-site Pauli matrix. In particular, we find a direct first-order phase transition (AFM-FM) for $(a,d_+)$ orbitals, while an intermediate phase for $(b,p_-)$ orbitals at $2.1\lesssim U_0/t\lesssim 2.3$. The intermediate phase is characterized by non-zero magnetism and staggered chirality, indicating specific magnetic pattern. As we argued in the main text, this phase belongs to the umbrella-type noncoplanar phase by canting the $\text{120}^o$ Neel configuration in the $x$-$y$ plane toward $z$ direction, from which the chirality is automatically staggered. 
 
Further decreasing the interaction strength goes beyond our strong coupling framework, and the electrons fail to arrange into localized orbitals under weak repel interaction.

\begin{figure}[htbp]

\includegraphics[width=0.8\linewidth]{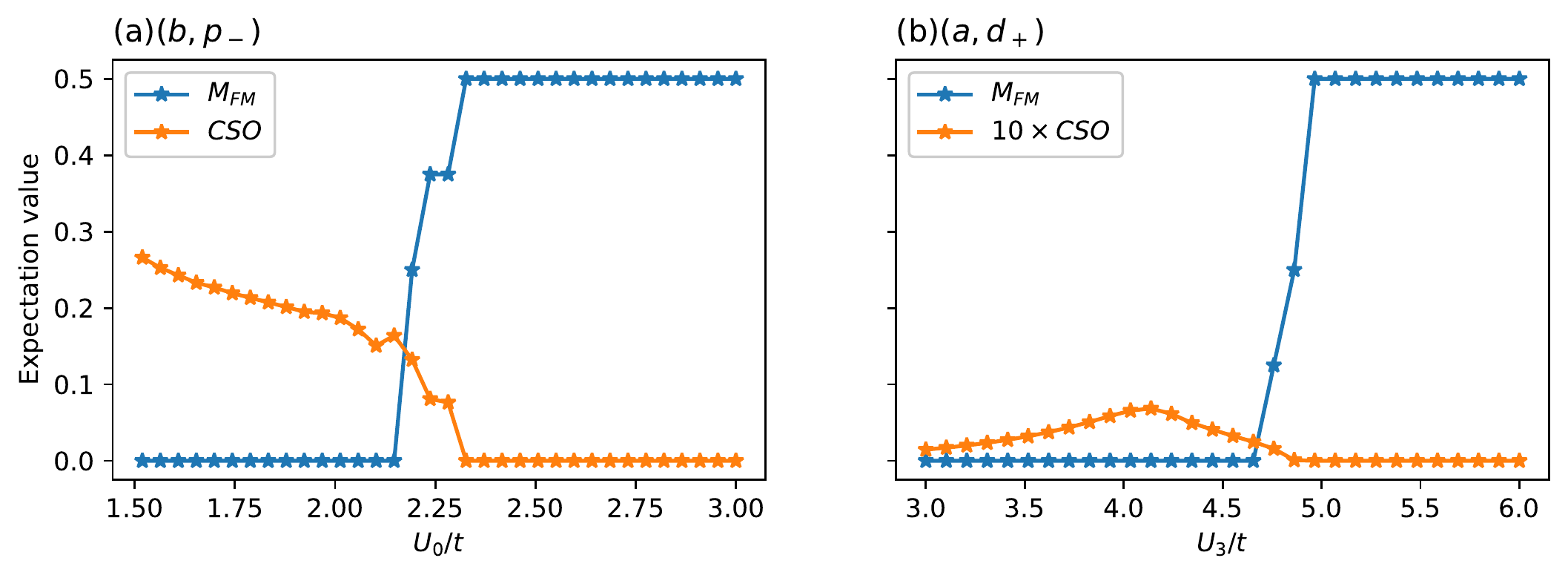}
\caption{Phase transition for the stably topological (bottom) band with respect to on-site interaction $U_0$ (a) and 3rd-neighbor interaction $U_3$ (b).
In (b), we multiply the $CSO$ by 10 to fit the graph. The $CSO$ is staggered for (a) (not shown here). }
\label{bottpt}
\end{figure}

In the main text, we present the localized orbitals of the Chern band. Now, we show the ones for fragile topological bands (combining top and bottom bands). To fill both of the bands, we need to fill two electrons per unit cell. Then the orbital index $I$ labels both the unit cell and the orbital in the unit cell. In the case of two electrons per unit cell, the strong repulsion between two electrons force them to form nematic orbitals localized on a single site as shown in Fig.\ref{appD}.

\begin{figure}[htbp]
 \includegraphics[width=.4\linewidth]{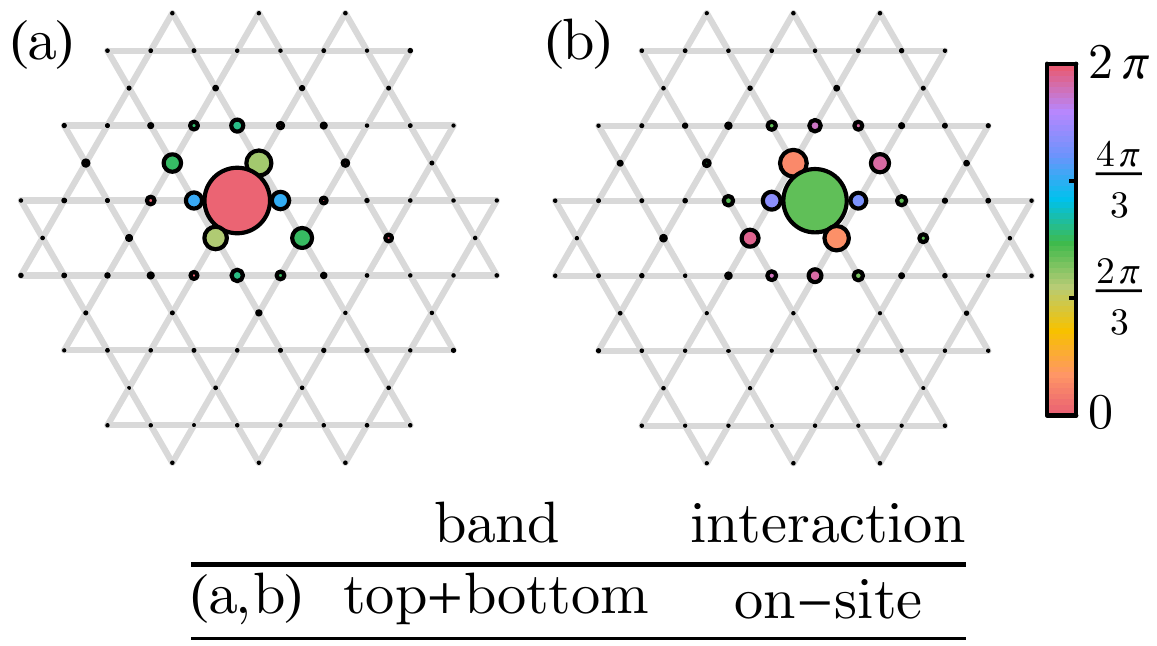}
\caption{Nonorthogonal Wannier orbitals for fragile topological band (combining top and bottom bands). (a) and (b) are a possible set of orbitals in a single unit cell. }
\label{appD}
\end{figure}

\end{document}